\documentclass[12pt,notitlepage]{amsart}%
\usepackage{amssymb}
\usepackage{amsfonts}
\usepackage{graphicx}
\usepackage{amscd}
\usepackage{graphicx}
\usepackage{amsmath}%
\newtheorem{theorem}{Theorem}
\theoremstyle{plain}
\newtheorem{acknowledgement}{Acknowledgement}

\newtheorem{definition}{Definition}

\newtheorem{lemma}{Lemma}

\newtheorem{proposition}{Proposition}
\newtheorem{remark}{Remark}

\numberwithin{equation}{section}
\numberwithin{theorem}{section}
\numberwithin{lemma}{section}
\numberwithin{proposition}{section}
\numberwithin{corollary}{section}

\ifx\pdfoutput\relax\let\pdfoutput=\undefined\fi
\newcount\msipdfoutput
\ifx\pdfoutput\undefined\else
\ifcase\pdfoutput\else
\ifx\paperwidth\undefined\else
\ifdim\paperheight=0pt\relax\else\pdfpageheight\paperheight\fi
\ifdim\paperwidth=0pt\relax\else\pdfpagewidth\paperwidth\fi
\fi\fi\fi
\begin{document}
\title[Koba-Nielsen String Amplitudes]{Meromorphic Continuation of Koba-Nielsen String Amplitudes}
\author{M. Bocardo-Gaspar}
\address{Universidad de Guadalajara, Cucei, Departamento de Matem\'{a}ticas, Blvd.
Marcelino Garc\'{\i}a Barrag\'{a}n \#1421,~Guadalajara, Jal. 44430, M\'{e}xico}
\email{miriam.bocardo@academicos.udg.mx.}
\author{Willem Veys}
\address{KU Leuven, Department of Mathematics\\
Celestijnenlaan 200 B, B-3001 Leuven, Belgium}
\email{wim.veys@kuleuven.be}
\author{W. A. Z\'{u}\~{n}iga-Galindo}
\address{Centro de Investigaci\'{o}n y de Estudios Avanzados del Instituto
Polit\'{e}cnico Nacional\\
Departamento de Matem\'{a}ticas, Unidad Quer\'{e}taro\\
Libramiento Norponiente \#2000, Fracc. Real de Juriquilla. Santiago de
Quer\'{e}taro, Qro. 76230\\
M\'{e}xico.}
\email{wazuniga@math.cinvestav.edu.mx}
\thanks{The second author was supported by KU Leuven grant C14/17/083. The third
author was partially supported by Conacyt Grant No. 250845.}

\begin{abstract}
In this article, we establish in a rigorous mathematical way that Koba-Nielsen
amplitudes defined on any local field of characteristic zero are bona fide
integrals that admit meromorphic continuations in the kinematic parameters.
Our approach allows us to study in a uniform way open and closed Koba-Nielsen
amplitudes over arbitrary local fields of characteristic zero. In the
regularization process we use techniques of local zeta functions and embedded
resolution of singularities. As an application we present the regularization
of $p$-adic open string amplitudes with Chan-Paton factors and constant
$B$-field. Finally, all the local zeta functions studied here are partition
functions of certain $1D$ log-Coulomb gases, which\ shows an interesting
connection between Koba-Nielsen amplitudes and statistical mechanics.

\end{abstract}
\maketitle

\section{Introduction}

In the recent years, scattering amplitudes, considered as mathematical
structures, have been studied intensively, see e.g. \cite{Arkani et al},
\cite{Elvang et al} and the references therein. The main motivations driving
this research are, from one side, the development of more efficient methods to
calculate amplitudes, and on the other side, the existence of deep connections
with many mathematical areas, among them, algebraic geometry, combinatorics,
number theory, $p$-adic analysis, etc., see e.g. \cite{Belkale},
\cite{Zun-B-C-LMP}-\cite{Zun-B-C-JHEP}, \cite{Bogner-Weiz-1}%
-\cite{Brown-Dupont-1}, \cite{Fra-Okada}-\cite{Zuniga-Compean et al},
\cite{Ghoshal:2004ay}-\cite{Grange:2004xj}, \cite{Lerner}, \cite{Marcolli}%
-\cite{Vanhove et al}, and the references therein. In the 80s the idea that
string amplitudes at the tree level can be studied over different number
fields and that there are connections between these amplitudes emerged in the
works of Freund, Witten and Volovich, among others, see e.g. \cite{B-F-O-W},
\cite{Brekke et al}, \cite{Vol}. In this framework the connection with local
zeta functions appears naturally. The present work is framed in the `emerging
idea' that scattering amplitudes are local zeta functions in the sense of
Gel'fand, Weil, Igusa, Sato, Bernstein, Denef, Loeser, etc., and also it
continues our investigation of the connections between string amplitudes at
the tree level and local zeta functions \cite{Zun-B-C-LMP},
\cite{Zun-B-C-JHEP}, \cite{Zuniga-Compean et al}.

In this article we establish, in a rigorous mathematical way, that the
Koba-Nielsen string amplitudes defined on any local field of characteristic
zero are bona fide integrals. Furthermore, they admit extensions which are
meromorphic complex functions in the kinematic parameters. We express the
Koba-Nielsen amplitudes as linear combinations of multivariate local zeta
functions, and, by using embedded resolution of singularities (Hironaka's
theorem \cite{H}), we show that all these local zeta functions are holomorphic
in a common domain, and then we use the fact that the local zeta functions
admit meromorphic continuations. Since Hironaka's theorem is valid over any
field of characteristic zero, we are able to regularize the Koba-Nielsen
amplitudes defined over $\mathbb{R}$, $\mathbb{C}$, or $\mathbb{Q}_{p}$, the
field of $p$-adic numbers, at the same time.

\smallskip We denote by $\mathbb{K}$ a local field of characteristic zero, and
set $\boldsymbol{f}:=\left(  f_{1},\ldots,f_{m}\right)  $ and $\boldsymbol{s}%
:=\left(  s_{1},\ldots,s_{m}\right)  \in\mathbb{C}^{m}$, where the $f_{i}(x)$
are non-constant polynomials in the variables $x:=(x_{1},\ldots,x_{n})$ with
coefficients in $\mathbb{K}$. The multivariate local zeta function attached to
$(\boldsymbol{f},\Theta)$, where $\Theta$ is a test function, is defined as
\[
Z_{\Theta}\left(  \boldsymbol{f},\boldsymbol{s}\right)  =\int
\limits_{\mathbb{K}^{n}}\Theta\left(  x\right)  \prod\limits_{i=1}%
^{m}\left\vert f_{i}(x)\right\vert _{\mathbb{K}}^{s_{i}}%
{\displaystyle\prod\limits_{i=1}^{n}}
dx_{i},\qquad\text{ when }\operatorname{Re}(s_{i})>0\text{ for all }i\text{,}%
\]
and where $%
{\textstyle\prod\nolimits_{i=1}^{n}}
dx_{i}$ is the normalized Haar measure of $(\mathbb{K}^{n},+)$. These
integrals admit meromorphic continuations to the whole $\mathbb{C}^{m}$,
\cite{Igusa-old}, \cite{Igusa}, \cite{Loeser}, see also \cite{G-S},
\cite{Kashiwara-Takai}. In the 60s, Weil studied local zeta functions, in the
Archimedean and non-Archimedean settings, in connection with the
Poisson-Siegel formula. In the 70s, Igusa developed a uniform theory for local
zeta functions in characteristic zero \cite{Igusa-old}, \cite{Igusa}, see also
\cite{Loeser}, \cite{Zuniga-Veys}, \cite{Zuniga-Veys2}. In the $p$-adic
setting, the local zeta functions are connected with the number of solutions
of polynomial congruences mod $p^{l}$ and with exponential sums mod $p^{l}$
\cite{Denef}. More recently, Denef and Loeser introduced the motivic zeta
functions which constitute a vast generalization of $p$-adic local zeta
functions \cite{DL1}.

In the case $\mathbb{K}=\mathbb{R}$ and $m=1$, the local zeta functions were
introduced in the 50s by Gel'fand and Shilov. The main motivation was that the
meromorphic continuation of Archimedean local zeta functions implies the
existence of fundamental solutions (i.e. Green functions) for differential
operators with constant coefficients. This fact was established,
independently, by Atiyah \cite{Atiyah} and Bernstein \cite{Ber}. It is
important to mention here that, in the $p$-adic framework, the existence of
fundamental solutions for pseudodifferential operators is also a consequence
of the fact that the Igusa local zeta functions admit a meromorphic
continuation, see \cite[Chapter 5]{Zuniga-LNM-2016}, \cite[Chapter
10]{KKZuniga}. This analogy turns out to be very important in the rigorous
construction of quantum scalar fields in the $p$-adic setting, see
\cite{M-V-Zuniga} and the references therein.

The connections between Feynman amplitudes and local zeta functions are very
old and deep. Let us mention that the works of Speer \cite{Speer} and Bollini,
Giambiagi and Gonz\'{a}lez Dom\'{\i}nguez \cite{B-G-Gonzalez-Dom} on
regularization of Feynman amplitudes in quantum field theory are based on the
analytic continuation of distributions attached to complex powers of
polynomial functions in the sense of Gel'fand and Shilov \cite{G-S}, see also
\cite{Belkale}, \cite{Bleher}, \cite{Bogner}, \cite{Marcolli}, among others.
The book \cite{G-S}, which is one of the main sources for the `$i\epsilon$
regularization method' widely used, was written before the establishing of
Hironaka's theorem \cite{H}. After the work of Atiyah, Bernstein and Igusa,
among others, the $i\epsilon$ regularization method was substituted by the
embedded resolution of singularities technique, see \cite{Igusa-old},
\cite{Igusa}. However, this method is not widely used by theoretical
physicists. In \cite{Witten} Witten discusses the classical $i\epsilon$
regularization method for string amplitudes; in this article, we present a
rigorous regularization of the Koba-Nielsen string amplitudes using the
`modern $i\epsilon$ regularization method'.

\smallskip Take $N\geq4$, and complex variables $s_{1j}$ and $s_{(N-1)j}$ for
$2\leq j\leq N-2$ and $s_{ij}$ for $2\leq i<j\leq N-2$. Put $\boldsymbol{s}%
:=\left(  s_{ij}\right)  \in\mathbb{C}^{\boldsymbol{d}}$, where
$\boldsymbol{d}=\frac{N(N-3)}{2}$ denotes the total number of indices $ij$. In
this article we introduce the multivariate local zeta function
\begin{equation}
Z_{\mathbb{K}}^{(N)}\left(  \boldsymbol{s}\right)  :=%
{\displaystyle\int\limits_{\mathbb{K}^{N-3}}}
{\displaystyle\prod\limits_{i=2}^{N-2}}
\left\vert x_{j}\right\vert _{\mathbb{K}}^{s_{1j}}\left\vert 1-x_{j}%
\right\vert _{\mathbb{K}}^{s_{(N-1)j}}\text{ }%
{\displaystyle\prod\limits_{2\leq i<j\leq N-2}}
\left\vert x_{i}-x_{j}\right\vert _{\mathbb{K}}^{s_{ij}}%
{\displaystyle\prod\limits_{i=2}^{N-2}}
dx_{i}, \label{zeta_funtion_string}%
\end{equation}
where ${\prod\nolimits_{i=2}^{N-2}}dx_{i}$ is the normalized Haar measure on
$\mathbb{K}^{N-3}$. We have called integrals of type
(\ref{zeta_funtion_string}) \textit{Koba-Nielsen local zeta functions}. These
functions have a statistical mechanics interpretation as partition functions
of certain $1D$ log-Coulomb gases, see Section
\ref{Section_partition_functions}.

We show that these functions are bona fide integrals, which are convergent and
holomorphic in an open part of $\mathbb{C}^{\mathbf{d}}$, containing the set
given by $\frac{-2}{N-2}<\operatorname{Re}(s_{ij})<\frac{-2}{N}$ for all $ij$.
Furthermore, they admit meromorphic continuations to the whole $\mathbb{C}%
^{\boldsymbol{d}}$, see Theorems \ref{TheoremA} and \ref{TheoremB}. We give a
detailed proof and a precise description of the convergence domain in the case
$\mathbb{K}=\mathbb{R}$, see Theorem \ref{TheoremB} and Section
\ref{Section_Road_map}; this proof can be easily extended to an arbitrary
local field $\mathbb{K}$ of characteristic zero, see Section
\ref{Section_general_case}.

The Koba-Nielsen open string amplitudes for $N$-points over $\mathbb{K}$ are
\textit{formally} defined as
\begin{equation}
A_{\mathbb{K}}^{(N)}\left(  \boldsymbol{k}\right)  :=%
{\displaystyle\int\limits_{\mathbb{K}^{N-3}}}
{\displaystyle\prod\limits_{i=2}^{N-2}}
\left\vert x_{j}\right\vert _{\mathbb{K}}^{\boldsymbol{k}_{1}\boldsymbol{k}%
_{j}}\left\vert 1-x_{j}\right\vert _{\mathbb{K}}^{\boldsymbol{k}%
_{N-1}\boldsymbol{k}_{j}}\text{ }%
{\displaystyle\prod\limits_{2\leq i<j\leq N-2}}
\left\vert x_{i}-x_{j}\right\vert _{\mathbb{K}}^{\boldsymbol{k}_{i}%
\boldsymbol{k}_{j}}%
{\displaystyle\prod\limits_{i=2}^{N-2}}
dx_{i}\text{,} \label{Amplitude}%
\end{equation}
where $\boldsymbol{k}=\left(  \boldsymbol{k}_{1},\ldots,\boldsymbol{k}%
_{N}\right)  $, $\boldsymbol{k}_{i}=\left(  k_{0,i},\ldots,k_{l,i}\right)
\in\mathbb{R}^{l+1}$, for $i=1,\ldots,N$ ($N\geq4$), is the momentum vector of
the $i$-th tachyon (with Minkowski product $\boldsymbol{k}_{i}\boldsymbol{k}%
_{j}=-k_{0,i}k_{0,j}+k_{1,i}k_{1,j}+\cdots+k_{l,i}k_{l,j}$), obeying
\begin{equation}
\sum_{i=1}^{N}\boldsymbol{k}_{i}=\boldsymbol{0}\text{, \ \ \ \ \ }%
\boldsymbol{k}_{i}\boldsymbol{k}_{i}=2\text{ \ for }i=1,\ldots,N.
\label{momenta-conservation}%
\end{equation}
The parameter $l$ is an arbitrary positive integer. Typically $l$ is taken to
be $25$. However, we do not require using the critical dimension. We choose
units such that the tachyon mass is $m^{2}=-2$. To carry out a mathematical
study of the integral (\ref{Amplitude}) it is more convenient to take
$\boldsymbol{k}_{i}\boldsymbol{k}_{j}=s_{ij}\in\mathbb{C}$. Several different
problems occur depending whether or not the $\boldsymbol{k}_{i}$ are real or
complex and if the kinematic restrictions (\ref{momenta-conservation}) are
considered or not.

The amplitudes of type (\ref{Amplitude}) were introduced in the works of
Brekke, Freund, Olson and Witten, among others, on string theory in the adelic
framework, see e.g. \cite[ Section 8]{Brekke et al}. We use here all the
conventions introduced in \cite{Brekke et al}. In the real case,
$A_{\mathbb{R}}^{(N)}\left(  \boldsymbol{k}\right)  $ is (up to multiplication
by a positive constant) the open Koba-Nielsen amplitude of $N$-points, see
\cite[ Section 8]{Brekke et al}, \cite[Section 2]{Kawai et al}. If $N=4$,
$A_{\mathbb{R}}^{(4)}\left(  \boldsymbol{k}\right)  $ is the Veneziano
amplitude \cite{Veneziano}. In the complex case, $A_{\mathbb{C}}^{(N)}\left(
\boldsymbol{k}\right)  $ is just a mathematical object. However, by using the
results of \cite[Section 2]{Kawai et al}, see also \cite{Blumenhagen et al},
the $N$-point, closed string amplitude at the tree level is the product of
$A_{\mathbb{C}}^{(N)}\left(  \boldsymbol{k}\right)  $ times a polynomial
function in the momenta $\boldsymbol{k}$. This fact implies that our
techniques and results are applicable to classical closed string amplitudes at
the tree level.

A central problem is to know whether or not integrals of type (\ref{Amplitude}%
) converge for some values of $\boldsymbol{k}$. We use the integrals
$Z_{\mathbb{K}}^{(N)}(\boldsymbol{s})$ as regularizations of the amplitudes
$A_{\mathbb{K}}^{(N)}\left(  \boldsymbol{k}\right)  $. More precisely, we
\textit{redefine}
\[
A_{\mathbb{K}}^{(N)}\left(  \boldsymbol{k}\right)  =Z_{\mathbb{K}}%
^{(N)}(\boldsymbol{s})\mid_{s_{ij}=\boldsymbol{k}_{i}\boldsymbol{k}_{j}},
\]
where $Z_{\mathbb{K}}^{(N)}(\boldsymbol{s})$ now denotes the meromorphic
continuation of (\ref{zeta_funtion_string}) to the whole $\mathbb{C}%
^{\boldsymbol{d}}$, see Theorem \ref{TheoremA}.

We show that $A_{\mathbb{K}}^{(N)}\left(  \boldsymbol{k}\right)  $ converges
on some open in $\mathbb{C}^{(N-1)(l+1)}$ by showing that this open is mapped
into the domain of convergency of $Z_{\mathbb{K}}^{(N)}\left(  \boldsymbol{s}%
\right)  $ by $\boldsymbol{k}\rightarrow s_{ij}=\boldsymbol{k}_{i}%
\boldsymbol{k}_{j}$. Our Theorem \ref{TheoremC} establishes further that
$A_{\mathbb{K}}^{(N)}\left(  \boldsymbol{k}\right)  $ extends to a meromorphic
function to the whole $\mathbb{C}^{(N-1)(l+1)}$, and that its polar set is
contained in the inverse image of the polar set of $Z_{\mathbb{K}}%
^{(N)}\left(  \boldsymbol{s}\right)  $ under that mapping. Also, it describes
the possible poles using numerical data of suitable resolutions of
singularities. It is important to mention here that, in the regularization of
$A_{\mathbb{K}}^{(N)}\left(  \boldsymbol{k}\right)  $, we do not use the
kinematic restrictions (\ref{momenta-conservation}). On the other hand, in the
cases $\mathbb{K}=\mathbb{R}$, $\mathbb{C}$, the meromorphic continuation of
$A_{\mathbb{K}}^{(N)}\left(  \boldsymbol{k}\right)  $ can be given in terms of
gamma functions, see Section \ref{Section_A_N_sum_gammas}.

As an illustration, we describe in the cases $N=4$, $5$, $6$ explicitly the
convergence domain of $Z_{\mathbb{K}}^{(N)}\left(  \boldsymbol{s}\right)  $,
see Examples \ref{Example_N_4}, \ref{Example_N_5}, \ref{Example_N_6}. In the
case $N=4$, $A_{\mathbb{R}}^{(4)}\left(  \boldsymbol{k}\right)  $ is exactly
the Veneziano amplitude \cite{Veneziano}. It is well-known that $A_{\mathbb{R}%
}^{(4)}\left(  \boldsymbol{k}\right)  $ can be expressed as a sum of gamma
functions, and the domain of convergence of the integral $A_{\mathbb{R}}%
^{(4)}\left(  \boldsymbol{k}\right)  $ is known. Our Theorem \ref{TheoremC}
gives exactly this domain of convergence, see Section
\ref{Section_Veneziaano_amplitude}.

Our methods allow in principle to express $A_{\mathbb{R}}^{(N)}\left(
\boldsymbol{k}\right)  $ as a sum of monomial integrals and as a linear
combination of gamma functions, with coefficients in the algebra of
holomorphic functions on $\mathbb{C}^{(N-1)(l+1)}$, see Section
\ref{Section_A_N_sum_gammas}. By using the result of Example \ref{Example_N_5}
(and \ref{Example_N_6}), this could be carried out explicitly for $N=5$ (and
$N=6$).

\smallskip The study of the scattering process of $N$-tachyons with momenta
$\boldsymbol{k}_{1},\ldots,\boldsymbol{k}_{N}\in\mathbb{R}^{l+1}$, each of
them with mass $m^{2}=-2$, requires determining a value of the meromorphic
function $A_{\mathbb{K}}^{(N)}\left(  \boldsymbol{k}\right)  $ for
$\boldsymbol{k}$ belonging to the real algebraic set
(\ref{momenta-conservation}). In this framework, $\boldsymbol{k}%
_{i}\boldsymbol{k}_{i}=2$ is the relativistic energy for the $i$-th tachyon.
Then, increasing $N$ means to increase the energy of the scattering process.
Our efforts for finding real solutions of (\ref{momenta-conservation})
contained in the domain of convergence of the integral $A_{\mathbb{K}}%
^{(N)}\left(  \boldsymbol{k}\right)  $ suggest that it is unlikely to find
such solutions for large $N$. In contrast, if $N\leq l+1$, we show that the
restriction of $A_{\mathbb{K}}^{(N)}\left(  \boldsymbol{k}\right)  $ to the
set (\ref{momenta-conservation}) gives rise to a well-defined scattering
amplitude, see Proposition \ref{Prop3}.

On the other hand, for $N>l+1$, the existence of points of the algebraic set
(\ref{momenta-conservation}) that are \emph{not} contained in the polar locus
of $A_{\mathbb{K}}^{(N)}\left(  \boldsymbol{k}\right)  $ is a non-trivial
question. For instance, in the basic case $N=4$ and $l=1$, such points do not
exist. As soon as $l \geq2$, we do expect that complex solutions exist, and we
constructed some for infinitely many $N$, see Example \ref{Lastexample}.

In the case $A_{\mathbb{K}}^{(N)}\left(  \boldsymbol{k}\right)  =\infty$, the
determination of a scattering amplitude requires a renormalization process for
the meromorphic amplitude $A_{\mathbb{K}}^{(N)}\left(  \boldsymbol{k}\right)
$. One can extend the function on the part of (\ref{momenta-conservation})
outside the polar locus, or proceed via the computation of Laurent series,
which are available in this case, see Section \ref{Section_A_N_sum_gammas}.

\smallskip Our results on Koba-Nielsen local zeta functions can be easily
extended to the case in which the integrals defining them contain
multiplicative characters. In Section \ref{Section-Chan_patton}, we show that
the techniques introduced here allow us to regularize general $p$-adic open
string amplitudes with Chan-Paton rules and a constant $B$-field. The
regularization problem of these integrals was posed in \cite{Zuniga-Compean et
al}, see also \cite{Ghoshal:2004ay}.

The string amplitudes were introduced by Veneziano in the 60s,
\cite{Veneziano}, further generalizations were obtained by Virasoro
\cite{Virasoro}, Koba and Nielsen \cite{Koba-Nielsen}, among others. The
$p$-adic string amplitudes emerged in the 80s in the works of Freund and Olson
\cite{F-O}, Freund and Witten \cite{F-W}, see also \cite{B-F-O-W}, Frampton
and Okada \cite{Fra-Okada}, and Volovich \cite{Vol}. Since the 60s the string
amplitudes at the tree level have been used as formal objects in many physical
calculations. In \cite{Zun-B-C-LMP}, it was established in the $p$-adic
setting and by using techniques of Igusa's local zeta functions that the
Koba-Nielsen amplitudes are bona fide integrals. In this article this result
is extended to an arbitrary local field of characteristic zero. We show that
the open and closed string amplitudes at the tree level can be studied in a
uniform way on any local field of characteristic zero, see Theorems
\ref{TheoremA}, \ref{TheoremC}. This is consistent with Volovich's conjecture
asserting that the mathematical description of physical reality must not
depend on the background number field, see \cite{Volovich1}.

\section{Discussion of the results}

After the publication of this manuscript in arxiv.org, we became aware of two
other manuscripts, also available in arxiv.org, containing some similar
results, see \cite{Brown-Dupont-1}, \cite{Vanhove et al}. Indeed, we are
dealing here with some matters considered in the mentioned works, but our
perspective, problems, and techniques are completely different, and mainly the
main results of all these works are complementary. Here we study string
amplitudes at the tree level as algebraic-arithmetic objects over arbitrary
local fields of characteristic zero. This implies that we can not use
presentations of the amplitudes that involve the standard real order. In
addition, we do not use the Mandelstam variables. In contrast, in
\cite{Brown-Dupont-1}, \cite{Vanhove et al}, the classical string amplitudes
are studied using moduli algebraic-geometric techniques. Then the results
about the convergence of the amplitudes (see \cite[Propositions 3.5 and
3.6]{Brown-Dupont-1}, \cite[Proposition 7.2]{Vanhove et al}) presented in
these works cannot be compared directly with our results, see e.g.
Propositions \ref{Lemma_Convergence}, \ref{Prop3}.

Our Theorem \ref{TheoremC} establishes the existence of a meromorphic
continuation for $A_{\mathbb{K}}^{(N)}\left(  \boldsymbol{k}\right)  $ in the
kinematic parameters. The meromorphic continuation provides more than a
regularization of the original amplitude; it shows in particular that the
original integral converges on some open region. The polar set of
$A_{\mathbb{K}}^{(N)}\left(  \boldsymbol{k}\right)  $ contains all the
information on the ultraviolet and infrared divergencies of the amplitude.
Furthermore, in the case $\mathbb{K}=\mathbb{R}$, $\mathbb{C}$, we show that
$A_{\mathbb{K}}^{(N)}\left(  \boldsymbol{k}\right)  $ is a linear combination
of gamma functions with coefficients in the ring of holomorphic functions on
$\mathbb{C}^{(N-1)(l+1)}$. The above-mentioned works do not contain similar
results. Theorem \ref{TheoremC} also implies\ the existence of Laurent
expansions for the string amplitudes in the kinematic parameters.

On the other hand, the main result of \cite{Brown-Dupont-1} is a
regularization of both open and closed string perturbation amplitudes at tree
level. Furthermore, the amplitudes admit Laurent expansion in Mandelstam
variables whose coefficients are multiple zeta values (resp. single-valued
multiple zeta values). Then, in \cite[Theorems 4.20 and 4.24]{Brown-Dupont-1}
a very general theory of Laurent series for string amplitudes is presented,
which gives a detailed local description of the string amplitudes. In
contrast, here we provide a global description of the open/closed string
amplitudes\ as meromorphic functions, but we do not provide a detailed
description of the coefficients of the Laurent expansions.

In Section \ref{Vanhove and Zerbini} we discuss a result of Vanhove and
Zerbini \cite[Proposition 7.2]{Vanhove et al} on the convergence of the
integrals $Z_{\mathbb{C}}^{\left(  N\right)  }(\boldsymbol{s})$. The authors
use an ad hoc procedure (partially also using changes of variables of \lq
blow-up type\rq) to study the convergence of these integrals. However, we
claim that their stated convergence domain is too large, in the sense that
some necessary inequalities were forgotten. We illustrate this in particular
for the case $N=5$, giving a concrete $\boldsymbol{s}_{0}$ in the stated
convergence domain of \cite[Proposition 7.2]{Vanhove et al} for which
$Z_{\mathbb{C}}^{\left(  5\right)  }(\boldsymbol{s}_{0})$ diverges. In our
view, this confirms that the technique of resolution of singularities is a
very appropriate tool to study convergence of Koba-Nielsen string amplitudes.

We also consider the problem of determining the scattering amplitude for the
interaction of $N$ tachyons, cf. Proposition \ref{Prop3}, which is not
discussed in \cite{Brown-Dupont-1}, \cite{Vanhove et al}. Finally, our
techniques can be used to study $p$-adic open string amplitudes, with
Chan-Paton rules and a constant $B$-field, and the local zeta functions
introduced here have a statistical mechanics interpretation. In the next two
sections we discuss these last matters.

\bigskip

\subsection{\label{Section-Chan_patton}Open string tree amplitudes with
Chan-Paton factors}

The connections between noncommutative geometry and string theory are very
deep and relevant, at mathematical and physical levels. The study of the
`Koba-Nielsen string amplitudes' coming from these theories is a relevant
matter. The mathematical framework developed here allows us naturally to study
these amplitudes. The study of such amplitudes is out of scope of
\cite{Brown-Dupont-1}, \cite{Vanhove et al}.

\bigskip

In ordinary string theory, the effective action for bosonic open strings in
gauge field backgrounds was discussed many years ago in
\cite{Abouelsaood:1986gd}. The analysis incorporating a Neveu-Schwarz
$B$-field in the target space leads to a noncommutative effective gauge theory
on the world-volume of D-branes \cite{Seiberg:1999vs}. The study of the
$p$-adic open string tree amplitudes including Chan-Paton factors was started
in \cite{B-F-O-W}. However, the incorporation of a $B$-field in the $p$-adic
context and the computation of the tree level string amplitudes was discussed
in \cite{Ghoshal:2004ay}, \cite{Grange:2004xj}. In these works it was reported
that the tree-level string amplitudes are affected by a noncommutative factor.
In \cite{Ghoshal:2004ay} Ghoshal and Kawano introduced new amplitudes
involving multiplicative characters and a noncommutative factor. These
amplitudes coincide with the ones obtained directly from the noncommutative
effective action \cite{Ghoshal:2004dd}.

\bigskip

In \cite{Zuniga-Compean et al} the regularization of the $p$-adic open string
amplitudes, with Chan-Paton rules and a constant $B$-field, introduced by
Ghoshal and Kawano, was established rigorously. The authors use techniques of
multivariate local zeta functions depending on multiplicative characters and a
phase factor which involves an antisymmetric bilinear form. By attaching to
each amplitude a multivariate local zeta function depending on the kinematic
parameters, the $B$-field and the Chan-Paton factors, the authors show that
these integrals admit meromorphic continuations in the kinematic parameters.

\bigskip

We define the $N$-point, open string amplitude, with Chan-Paton rules in a
constant $B$-field over $\mathbb{K}$, by
\begin{gather}
A_{\mathbb{K}}^{(N)}\left(  \boldsymbol{k},\theta,\mathrm{sgn}_{\mathbb{K}%
}\right)  =\label{NC-Amplitude}\\
{\int\limits_{\mathbb{K}^{N}}}\text{ }{\prod\limits_{1\leq i<j\leq N}%
}\left\vert x_{i}-x_{j}\right\vert _{\mathbb{K}}^{\boldsymbol{k}%
_{i}\boldsymbol{k}_{j}}H_{\mathbb{K}}\left(  x_{i}-x_{j}\right)  \text{ }%
\exp\left\{  -\frac{\sqrt{-1}}{2}\left(  {\sum\limits_{1\leq i<j\leq N}%
}(\boldsymbol{k}_{i}\theta\boldsymbol{k}_{j})\mathrm{sgn}_{\mathbb{K}}%
(x_{i}-x_{j})\right)  \right\}  {\prod\limits_{i=1}^{N}}dx_{i}\text{,}%
\nonumber
\end{gather}
where $N\geq4$, $\boldsymbol{k}=\left(  \boldsymbol{k}_{1},\ldots
,\boldsymbol{k}_{N}\right)  $, $\boldsymbol{k}_{i}=\left(  k_{0,i}%
,\ldots,k_{l,i}\right)  $, $i=1,\ldots,N$, is the momentum vector of the
$i$-th tachyon vertex operator obeying (\ref{momenta-conservation}),
$H_{\mathbb{K}}\left(  x\right)  =\frac{1}{2}(1+\mathrm{sgn}_{\mathbb{K}}%
(x))$, $\mathrm{sgn}_{\mathbb{K}}(x)$ is a $\mathbb{K}$-version of the sign
function, $\theta$ is a fixed antisymmetric bilinear form, and ${\prod
\nolimits_{i=1}^{N}}dx_{i}$ is the normalized Haar measure on $\left(
\mathbb{K}^{N},+\right)  $. These amplitudes are non-commutative
generalizations of the Koba-Nielsen amplitudes. Indeed, the non-commutativity
comes from the fact that $\boldsymbol{k}_{j}\theta\boldsymbol{k}_{i}%
\neq\boldsymbol{k}_{i}\theta\boldsymbol{k}_{j}$, and by turning of the
$B$-field, i.e., taking $\theta=0$ and $\mathrm{sgn}_{\mathbb{K}}=1$, one
obtains an amplitude similar to (\ref{Amplitude}) but in a different number of variables.

The functions $\mathrm{sgn}_{\mathbb{K}}$ are multiplicative characters of the
multiplicative group $\left(  \mathbb{K}^{\times},\cdot\right)  $. For any
$\mathbb{K}$ there is a trivial character, $\mathrm{sgn}_{\mathbb{K}}\left(
x\right)  =1$ for any $x\in\mathbb{K}^{\times}$. In the case $\mathbb{K}%
=\mathbb{R}$, there is only one non-trivial multiplicative character,
$\mathrm{sgn}_{\mathbb{R}}\left(  x\right)  =\frac{x}{\left\vert x\right\vert
}$, $x\in\mathbb{R}^{\times}$; in the case $\mathbb{K}=\mathbb{C}$, the
nontrivial multiplicative characters have the form $\mathrm{sgn}_{\mathbb{C}%
}\left(  x\right)  =\left(  \frac{x}{\left\vert x\right\vert }\right)  ^{m}$,
for some integer $m$, and $x\in\mathbb{C}^{\times}$. For the discussion of the
$p$-adic case the reader may consult \cite{Zuniga-Compean et al}.

This type of theories is not invariant under projective M\"{o}bius
transformations and consequently the normalization
\begin{equation}
x_{1}=0,x_{N-1}=1,x_{N}=\infty\label{gauge}%
\end{equation}
can not be carried out. In \cite{Ghoshal:2004ay}, in the case $\mathbb{K}%
=\mathbb{Q}_{p}$, the authors assume such a normalization and study the
corresponding amplitudes for some particular values of $N$. In
\cite{Zuniga-Compean et al}, these amplitudes were studied for an arbitrary
number of points, using resolution of singularities. One of the conclusions
obtained there was that imposing the conformal gauge (\ref{gauge}) produces
amplitudes with strange properties, and consequently, the study of amplitudes
of the form (\ref{NC-Amplitude}) is completely necessary. It is important to
mention that integrals (\ref{NC-Amplitude}) appear in a $p$-adic framework,
but since they are algebro-geometric objects, here we study them over
arbitrary fields of characteristic zero.

As an application of the techniques and results presented here, we show
implicitly also the existence of a meromorphic continuation for amplitudes of
type (\ref{NC-Amplitude}), in the case $\mathbb{K}$ is $\mathbb{R}$ or a
$p$-adic field. We attach to $A_{\mathbb{K}}^{(N)}\left(  \boldsymbol{k}%
,\theta,\mathrm{sgn}_{\mathbb{K}}\right)  $ the Igusa type integral%
\[
Z_{\mathbb{K}}^{(N)}\left(  \boldsymbol{s},\widetilde{\boldsymbol{s}%
},\mathrm{sgn}_{\mathbb{K}}\right)  :={\int\limits_{\mathbb{K}^{N}}}\text{
}{\prod\limits_{1\leq i<j\leq N}}\left\vert x_{i}-x_{j}\right\vert
_{\mathbb{K}}^{s_{ij}}H_{\mathbb{K}}\left(  x_{i}-x_{j}\right)  \text{
}E\left(  \boldsymbol{x},\widetilde{\boldsymbol{s}},\mathrm{sgn}_{\mathbb{K}%
}\right)  {\prod\limits_{i=1}^{N}}dx_{i},
\]
where%
\[
E\left(  \boldsymbol{x},\widetilde{\boldsymbol{s}},\mathrm{sgn}_{\mathbb{K}%
}\right)  :=\exp\left\{  -\frac{\sqrt{-1}}{2}\left(  {\sum\limits_{1\leq
i<j\leq N}}\widetilde{s}_{ij}\mathrm{sgn}_{\mathbb{K}}(x_{i}-x_{j})\right)
\right\}  ,
\]
the $s_{ij}$ are complex symmetric variables and the $\widetilde{s}_{ij}$ are
real antisymmetric variables.

If $\mathbb{K}$ is $\mathbb{R}$ or a $p$-adic field, \ then $\mathrm{sgn}%
_{\mathbb{K}}\left(  x\right)  \in\left\{  \pm1\right\}  $ for $x\in
\mathbb{K}^{\times}$, and then for $\widetilde{s}\in\mathbb{R}$ and
$x\in\mathbb{K}^{\times}$,%
\[
\exp\left(  -\frac{\sqrt{-1}\widetilde{s}}{2}\mathrm{sgn}_{\mathbb{K}}\left(
x\right)  \right)  =\cos\left(  \frac{\widetilde{s}}{2}\right)  -\sqrt
{-1}\mathrm{sgn}_{\mathbb{K}}\left(  x\right)  \sin\left(  \frac{\widetilde
{s}}{2}\right)  .
\]
By using this identity, we get%
\[
E\left(  \boldsymbol{x},\widetilde{\boldsymbol{s}},\mathrm{sgn}_{\mathbb{K}%
}\right)  ={\sum\limits_{J}}E_{J}(\widetilde{\boldsymbol{s}})\text{ }%
{\prod\limits_{i,j\in J;\text{ }i<j}}\mathrm{sgn}_{\mathbb{K}}(x_{i}-x_{j}),
\]
where the $E_{J}(\widetilde{\boldsymbol{s}})$ are polynomial functions in the
variables $\cos\left(  \widetilde{s}_{ij}\right)  $, $\sin\left(
\widetilde{s}_{ij}\right)  $. In addition, we also have the identity%
\[
{\prod\limits_{1\leq i<j\leq N}}H_{\mathbb{K}}(x_{i}-x_{j})=c+{\sum
\limits_{J}}e_{J}{\prod\limits_{i,j\in J;\text{ }i<j}}\mathrm{sgn}%
_{\mathbb{K}}(x_{i}-x_{j}),
\]
where $c$ and the $e_{J}$ are constants, and $J$ runs over a family of subsets
of $\left\{  1,\ldots,N\right\}  $, see also \cite[Section 5.1]{Zuniga-Compean
et al}. By using these identities, $Z_{\mathbb{K}}^{(N)}\left(  \boldsymbol{s}%
,\widetilde{\boldsymbol{s}},\mathrm{sgn}_{\mathbb{K}}\right)  $ is a linear
combination of twisted Koba-Nielsen zeta functions of the form
\[
Z_{\mathbb{K}}^{\left(  N\right)  }\left(  \boldsymbol{s},\mathbf{\chi
};J\right)  :={\int\limits_{\mathbb{K}^{N}}}\text{ }{\prod\limits_{i,j\in
J;\text{ }i<j}}\left\{  \left\vert x_{i}-x_{j}\right\vert _{\mathbb{K}%
}^{s_{ij}}\chi_{ij}(x_{i}-x_{j})\right\}  {\prod\limits_{i=1}^{N}}dx_{i},
\]
where $\mathbf{\chi=}\left(  \chi_{ij}\right)  _{i,j\in J}$ and $\chi
_{ij}=\mathrm{sgn}_{\mathbb{K}}$ or $\chi_{ij}=1$. By using a partition of the
unity of the form $1_{\mathbb{K}}=\sum_{I\subset\left\{  1,\ldots,N\right\}
}\Theta_{I}\left(  \boldsymbol{x}\right)  $, where $\Theta_{I}$ is a smooth
function (in the case $\mathbb{K}=\mathbb{R}$,) or a locally constant function
(in the non-Archimedean case), see e.g. Definition \ref{Definition_phi},
$Z_{\mathbb{K}}^{\left(  N\right)  }\left(  \boldsymbol{s},\mathbf{\chi
};J\right)  $ becomes a linear combination of integrals of the form%
\begin{equation}
Z_{\mathbb{K}}^{(N)}\left(  \boldsymbol{s},\mathbf{\chi};\Theta_{J},J\right)
:={\int\limits_{\mathbb{K}^{N}}}\text{ }\Theta_{J}\left(  \boldsymbol{x}%
\right)  \text{ }{\prod\limits_{i,j\in J;\text{ }i<j}}\left\{  \left\vert
x_{i}-x_{j}\right\vert _{\mathbb{K}}^{s_{ij}}\chi_{ij}(x_{i}-x_{j})\right\}
{\prod\limits_{i=1}^{N}}dx_{i}. \label{Integral_zeta}%
\end{equation}
These integrals are twisted multivariate local zeta functions if $\Theta_{J}$
has compact support.

By renaming the variables, we may assume without loss of generality that%
\[
{\prod\limits_{i,j\in J;\text{ }i < j}}\left\{  \left\vert x_{i}%
-x_{j}\right\vert _{\mathbb{K}}^{s_{ij}}\chi_{ij}(x_{i}-x_{j})\right\}
={\prod\limits_{l\leq i<j\leq n}}\left\{  \left\vert x_{i}-x_{j}\right\vert
_{\mathbb{K}}^{s_{ij}}\chi_{ij}(x_{i}-x_{j})\right\}  ,
\]
for some integers $l$, $n$ satisfying \ $1\leq l<n\leq N$. If $\Theta_{J}$ has
an unbounded support, we may assume that there exists $m$ satisfying $1\leq l
\leq m \leq n\leq N$, such that $x_{i}$ runs through a bounded set for
$i=l,\ldots,m-1$, and through a unbounded set for $i=m,\ldots,n$. By changing
variables $x_{i}\rightarrow x_{i}$ for $i=l,\ldots,m-1$ and $x_{i}\rightarrow
x_{i}^{-1}$ for $i=m,\ldots,n$, in (\ref{Integral_zeta}), we get%
\[
Z_{\mathbb{K}}^{(N)}\left(  \boldsymbol{s},\mathbf{\chi};\widetilde{\Theta
}_{J},J\right)  ={\int\limits_{\mathbb{K}^{N}}}\text{ }\widetilde{\Theta}%
_{J}\left(  \boldsymbol{x}\right)  \text{ }\frac{F_{J}\left(  x,\boldsymbol{s}%
,\boldsymbol{\chi}\right)  }{\prod\limits_{i=m}^{n}\left\vert x_{i}\right\vert
_{\mathbb{K}}^{2}}{\prod\limits_{i=1}^{N}}dx_{i},
\]
where%
\begin{align*}
F_{J}\left(  x,\boldsymbol{s},\boldsymbol{\chi}\right)   &  :={\prod
\limits_{\substack{l\leq i\leq m-1\\m\leq j\leq n}}}\left\{  \left\vert
x_{i}x_{j}-1\right\vert _{\mathbb{K}}^{s_{ij}}\chi_{ij}(x_{i}x_{j}%
-1)\left\vert x_{j}\right\vert _{\mathbb{K}}^{-s_{ij}}\chi_{ij}^{-1}%
(x_{j})\right\}  \times\\
&  {\prod\limits_{l\leq i<j\leq m-1}}\left\{  \left\vert x_{i}-x_{j}%
\right\vert _{\mathbb{K}}^{s_{ij}}\chi_{ij}(x_{i}-x_{j})\right\}  \times\\
&  {\prod\limits_{m\leq i<j\leq n}}\left\{  \left\vert x_{i}-x_{j}\right\vert
_{\mathbb{K}}^{s_{ij}}\chi_{ij}(x_{i}-x_{j})\left\vert x_{i}\right\vert
_{\mathbb{K}}^{-s_{ij}}\chi_{ij}^{-1}(x_{i})\left\vert x_{j}\right\vert
_{\mathbb{K}}^{-s_{ij}}\chi_{ij}^{-1}(x_{j})\right\}  .
\end{align*}
By using Hironaka's resolution of singularities theorem, one establishes the
existence of meromorphic continuations (regularizations) for each integral
$Z_{_{\mathbb{K}}}^{(N)}\left(  \boldsymbol{s},\mathbf{\chi};\widetilde
{\Theta}_{J},J\right)  $, see e.g. \cite{Loeser}.
Here we can show the existence of a common domain where all these integrals
converge.
In the cases $\mathbb{K}=\mathbb{R}$, $\mathbb{Q}_{p}$, the description of the
possible poles for the integrals $Z_{_{\mathbb{K}}}^{(N)}\left(
s,\mathbf{\chi};\widetilde{\Theta}_{J},J\right)  $ does not depend on
$\mathbf{\chi}$, see \cite[Theorems 5.4.1, 8.2.1]{Igusa}. Consequently, all
the results presented in Sections \ref{Section_Road_map} and
\ref{Section_general_case}, including Theorem \ref{TheoremA}, extend to these
integrals, by replacing $N$ by $N+3$.

For an arbitrary local field $\mathbb{K}$ of characteristic zero, by
considering
\[
\left\vert Z_{\mathbb{K}}^{\left(  N\right)  }\left(  \boldsymbol{s}%
,\mathbf{\chi}\right)  \right\vert \text{, }\left\vert Z_{\mathbb{K}}%
^{(N)}\left(  \boldsymbol{s},\mathbf{\chi};\widetilde{\Theta}_{J},J\right)
\right\vert ,
\]
and using Theorem \ref{TheoremB}, we conclude that all these integrals
converge (and are holomorphic) in an open domain, containing the set
\begin{equation}
-\frac{2}{N+1}<\operatorname{Re}(s_{ij})<-\frac{2}{N+3}\qquad\text{ for all
}ij\text{.} \label{Region}%
\end{equation}
And consequently $\left\vert Z_{\mathbb{K}}^{(N)}\left(  \boldsymbol{s}%
,\widetilde{\boldsymbol{s}},\mathrm{sgn}_{\mathbb{K}}\right)  \right\vert $
also converges in this neighborhood.

Finally, we regularize $A_{\mathbb{K}}^{(N)}\left(  \boldsymbol{k}%
,\theta,\mathrm{sgn}_{\mathbb{K}}\right)  $ by taking%
\[
A_{\mathbb{K}}^{(N)}\left(  \boldsymbol{k},\theta,\mathrm{sgn}_{\mathbb{K}%
}\right)  =Z_{\mathbb{K}}^{(N)}\left(  \boldsymbol{s},\widetilde
{\boldsymbol{s}},\mathrm{sgn}_{\mathbb{K}}\right)  \mid_{s_{ij}=\boldsymbol{k}%
_{i}\boldsymbol{k}_{j}\text{, }\widetilde{s_{ij}}=\boldsymbol{k}_{i}%
\theta\boldsymbol{k}_{j}}.
\]
The fact that $A_{\mathbb{K}}^{(N)}\left(  \boldsymbol{k},\theta
,\mathrm{sgn}_{\mathbb{K}}\right)  $ admits an extension which is meromorphic
in the kinematic parameters is established as in the proof of Theorem
\ref{TheoremC}.

\subsection{\label{Section_partition_functions} Koba-Nielsen local zeta
functions and $1D$ log-Coulomb gases}

The Koba-Nielsen local zeta functions are partition functions of certain $1D$
log-Coulomb gases. It is worth to mention that a similar connection has been
reported in conformal field theory, see e.g. \cite{Kawai-1}-\cite{Kawai-2},
see also \cite{Balaska}\ and the references therein. Let $\left(
\mathbb{K},\left\vert \cdot\right\vert _{\mathbb{K}}\right)  $ be a local
field of characteristic zero. The sites of a configuration of a log
Koba-Nielsen-Coulomb gas are vectors of the form $\left(  0,x_{2}%
,\ldots,x_{N-2},1\right)  \in\mathbb{K}^{N-1}$, which means that two sites are
fixed: $x_{1}=0$ and $x_{N-1}=1$. At the site $x_{i}$ there is a charge
$e_{i}\in\mathbb{R}$. We set $\boldsymbol{e}=\left(  e_{i}\right)  _{1\leq
i\leq N-1}$. The Hamiltonian of the log Koba-Nielsen-Coulomb gas is%
\begin{multline*}
H_{\mathbb{K}}(x_{2},\ldots,x_{N-2};\boldsymbol{e})=-\sum\limits_{i=2}%
^{N-2}\ln\left\vert x_{i}\right\vert _{\mathbb{K}}^{e_{i}e_{1}}-\sum
\limits_{i=2}^{N-2}\ln\left\vert 1-x_{i}\right\vert _{\mathbb{K}}%
^{e_{i}e_{(N-1)i}}\\
-\sum\limits_{2\leq i<j\leq N-2}\ln\left\vert x_{i}-x_{j}\right\vert
_{\mathbb{K}}^{e_{i}e_{j}}.
\end{multline*}
The statistical mechanics of this Coulomb gas is described by the
Gibbs\ measure
\[
\frac{e^{-\beta H_{\mathbb{K}}(x_{2},\ldots,x_{N-2};\boldsymbol{e})}%
}{\mathcal{Z}_{\mathbb{K},N,\boldsymbol{e}}\left(  \beta\right)  }%
{\displaystyle\prod\limits_{i=2}^{N-2}}
dx_{i},
\]
where $%
{\textstyle\prod\nolimits_{i=2}^{N-2}}
dx_{i}$ is a normalized Haar measure of the topological group $\left(
\mathbb{K}^{N-3},+\right)  $, $\beta$ $>0$ is the inverse temperature and
$\mathcal{Z}_{\mathbb{K},\beta,N,\boldsymbol{e}}$ is a normalization constant,
the partition function, defined as%
\[
\mathcal{Z}_{\mathbb{K},N,\boldsymbol{e}}\left(  \beta\right)  =%
{\displaystyle\int\limits_{\mathbb{K}^{N-3}}}
{\displaystyle\prod\limits_{i=2}^{N-2}}
\left\vert x_{i}\right\vert _{\mathbb{K}}^{\beta e_{i}e_{1}}\left\vert
1-x_{i}\right\vert _{\mathbb{K}}^{\beta e_{i}e_{(N-1)i}}%
{\displaystyle\prod\limits_{2\leq i<j\leq N-2}}
\left\vert x_{i}-x_{j}\right\vert _{\mathbb{K}}^{\beta e_{i}e_{j}}%
{\displaystyle\prod\limits_{i=2}^{N-2}}
dx_{i}.
\]
Then, Theorem \ref{TheoremA} implies that $\mathcal{Z}_{\mathbb{K}%
,N,\boldsymbol{e}}\left(  \beta\right)  $ converges in some region determined
by the $\beta e_{i}e_{j}$, and it admits a meromorphic continuation to the
whole $\mathbb{C}^{\frac{N(N-3)}{2}}$. It is interesting to mention that
positive poles of $\mathcal{Z}_{\mathbb{K},N,\boldsymbol{e}}\left(
\beta\right)  $ are related with phase transitions, see \cite{Zuniga et al}
and the references therein.

\section{\label{Section_3}Multivariate Local Zeta Functions and Embedded
Re\-solution of Singularities}

Let $f_{1}(x),\ldots,f_{m}(x)\in\mathbb{R}\left[  x_{1},\ldots,x_{n}\right]  $
be non-constant polynomials; we denote by $D:=\cup_{i=1}^{m}f_{i}^{-1}(0)$ the
divisor attached to them. We set $\boldsymbol{f}:=\left(  f_{1},\ldots
,f_{m}\right)  $ and $\boldsymbol{s}:=\left(  s_{1},\ldots,s_{m}\right)
\in\mathbb{C}^{m}$. For each $\Theta:\mathbb{R}^{n}\rightarrow\mathbb{C}$
smooth with compact support, the multivariate local zeta function attached to
$(\boldsymbol{f},\Theta)$ is defined as%
\begin{equation}
Z_{\Theta}\left(  \boldsymbol{f},\boldsymbol{s}\right)  :=\int
\limits_{\mathbb{R}^{n}\smallsetminus D}\Theta\left(  x\right)  \prod
\limits_{i=1}^{m}\left\vert f_{i}(x)\right\vert ^{s_{i}}%
{\displaystyle\prod\limits_{i=1}^{n}}
dx_{i} , \label{Zeta_Function}%
\end{equation}
when $\operatorname{Re}(s_{i})>0$ for all $i$. Integrals of type
(\ref{Zeta_Function}) are analytic functions, and they admit meromorphic
continuations to the whole $\mathbb{C}^{m}$, see \cite{Loeser}, \cite{Igusa},
\cite{Igusa-old}. By applying Hironaka's resolution of singularities theorem
to $D$, the study of integrals of type (\ref{Zeta_Function}) is reduced to the
case of monomials integrals, which can be studied directly, see e.g.
\cite{Loeser}, \cite{Igusa}, \cite{Igusa-old}.

\begin{theorem}
[Hironaka, \cite{H}]\label{thresolsing} There exists an embedded resolution
$\sigma:X\rightarrow\mathbb{R}^{n}$ of $D$, that is,

\noindent(i) $X$ is an $n$-dimensional $\mathbb{R}$-analytic manifold,
$\sigma$ is a proper $\mathbb{R}$-analytic map which is a composition of a
finite number of blow-ups at closed submanifolds, and which is an isomorphism
outside of $\sigma^{-1}(D)$;

\noindent(ii) $\sigma^{-1}\left(  D\right)  $ is a normal crossings divisor,
meaning that $\sigma^{-1}\left(  D\right)  =\cup_{i\in T}E_{i}$, where the
$E_{i}$\ are closed submanifolds of $X$ of codimension one, each equipped with
an $m$-tuple of nonnegative integers $\left(  N_{f_{1},i},\ldots,N_{f_{m}%
,i}\right)  $ and a positive integer $v_{i}$, satisfying the following. At
every point $b$ of $X$ there exist local coordinates $\left(  y_{1}%
,\ldots,y_{n}\right)  $ on $X$ around $b$ such that, if $E_{1},\ldots,E_{r}$
are the $E_{i}$ containing $b$, we have on some open neighborhood $V$ of $b$
that $E_{i}$ is given by $y_{i}=0$ for $i\in\{1,\ldots,r\}$,
\begin{equation}
\sigma^{\ast}\left(  dx_{1}\wedge\ldots\wedge dx_{n}\right)  =\eta\left(
y\right)  \left(  \prod_{i=1}^{r}y_{i}^{v_{i}-1}\right)  dy_{1}\wedge
\ldots\wedge dy_{n}, \label{for2}%
\end{equation}
and
\begin{equation}
f_{j}^{\ast}(y):=\left(  f_{j}\circ\sigma\right)  \left(  y\right)
=\varepsilon_{f_{j}}\left(  y\right)  \prod_{i=1}^{r}y_{i}^{N_{f_{j},i}}
\qquad\text{ for } j=1,\ldots, m , \label{for4}%
\end{equation}
where $\eta\left(  y\right)  $ and the $\varepsilon_{f_{j}}\left(  y\right)  $
belong to $\mathcal{O}_{X,b}^{\times}$, the group of units of the local ring
of $X$\ at $b$.
\end{theorem}

There are two kinds of submanifolds $E_{i},i\in T$. Each blow-up creates an
\textit{exceptional variety} $E_{i}$, the image by $\sigma$ of \ any of these
$E_{i}$ has codimension at least two in $\mathbb{R}^{n}$. The other $E_{i}$
are the so-called \textit{strict transforms} of the irreducible components of
$D$.

When using Hironaka's resolution theorem, we will identify the Lebesgue
measure $\prod\nolimits_{i=1}^{n}dx_{i}$ with the measure induced by the top
differential form $dx_{1}\wedge\ldots\wedge dx_{n}$ in $\mathbb{R}^{n}$. For a
discussion on the basic aspects of analytic manifolds and resolution of
singularities, the reader may consult \cite[Chapter 2]{Igusa}. More generally,
Hironaka's resolution theorem is valid over any field of characteristic zero,
in particular over the local fields $\mathbb{R}$, $\mathbb{C}$, the field of
$p$-adic numbers $\mathbb{Q}_{p}$, or a finite extension of $\mathbb{Q}_{p}$.

The resulting monomial integrals are then handled by the following lemma,
which is an easy variation of well-known results, see e.g. \cite[Chap. II,
\S \ 7, \ Lemme 4]{AVG}, \cite[Lemme 3.1]{D-S}, \cite[Chap. I, Sect. 3.2]%
{G-S}, and \cite[Lemma 4.5]{Igusa-old}.

\begin{lemma}
\label{Lemma0} Consider the integral
\[
J(s_{1},\ldots,s_{m})=%
{\textstyle\int\limits_{\mathbb{R}^{n}}}
\Phi\left(  y,s_{1},\ldots,s_{m}\right)  \prod_{i=1}^{r}\left\vert
y_{i}\right\vert ^{\sum_{j=1}^{m}a_{{j},i}s_{j}+b_{i}-1}%
{\displaystyle\prod\limits_{i=1}^{n}}
dy_{i}\text{, }%
\]
where $1\leq r\leq n$, for each $i$ the $a_{j,i}$ are integers (not all zero)
and $b_{i}$ is an integer, and $\Phi\left(  y,s_{1},\ldots,s_{m}\right)  $ is
a smooth function with support in the polydisc
\[
\left\{  y\in\mathbb{R}^{n};\left\vert y_{i}\right\vert <1\text{, for
}i=1,\ldots,n\right\}  ,
\]
which is holomorphic in $s_{1},\ldots,s_{m}$. Set%
\[
\mathcal{R}:=%
{\textstyle\bigcap\limits_{i\in\left\{  1,\ldots,r\right\}  }}
\left\{  (s_{1},\ldots,s_{m})\in\mathbb{C}^{m};\sum_{j=1}^{m}a_{j,i}%
\operatorname{Re}(s_{j})+b_{i}>0\right\}  .
\]
Then the following assertions hold:

\noindent(i) if all the $a_{j,i}$ are nonnegative integers (not all zero) and
$b_{i}$ is a positive integer, then $\mathcal{R}\neq\emptyset$. More
precisely, $\left\{  (s_{1},\ldots,s_{m})\in\mathbb{C}^{m};\operatorname{Re}%
(s_{j})>0\text{, }j=1,\cdots,m\right\}  \subset\mathcal{R}$;

\noindent(ii) if $\mathcal{R}\neq\emptyset$, then $J(s_{1},\ldots,s_{m})$ is
convergent and defines a holomorphic function in the domain $\mathcal{R}$;

\noindent(iii) if $\mathcal{R}\neq\emptyset$, $J(s_{1},\ldots,s_{m})$ admits
an analytic continuation to the whole $\mathbb{C}^{m}$,\ as a meromorphic
function with poles belonging to
\[
\bigcup_{1\leq i\leq r}\text{ }\bigcup_{t\in\mathbb{N}}\left\{  \sum_{j=1}%
^{m}a_{j,i}s_{j}+b_{i}+t=0\right\}  .
\]

\end{lemma}

\begin{remark}
\label{rem: holom} Integrals of the form $J(s_{1},\ldots,s_{m})$, where the
$a_{j,i}$ and $b_{i}$ are positive integers, have been extensively studied,
see e.g. \cite[Chap. II, \S \ 7, \ Lemme 4]{AVG}, \cite[Lemme 3.1]{D-S},
\cite[Chap. I, Sect. 3.2]{G-S}, and \cite[Lemma 4.5]{Igusa-old}.
\end{remark}

Combining Theorem \ref{thresolsing} and Lemma \ref{Lemma0}, the precise
conclusion is as follows.

\begin{theorem}
\label{thm: num data and poles} Let $f_{1}(x),\ldots,f_{m}(x)\in
\mathbb{R}\left[  x_{1},\ldots,x_{n}\right]  $ be non-constant polynomials and
$\Theta:\mathbb{R}^{n}\rightarrow\mathbb{C}$ a smooth function with compact
support, to which we associate the multivariate local zeta function
$Z_{\Theta}\left(  \boldsymbol{f},\boldsymbol{s}\right)  $ as in
(\ref{Zeta_Function}). Fix an embedded resolution $\sigma:X\rightarrow
\mathbb{R}$ of $D=\cup_{i=1}^{m}f_{i}^{-1}(0)$ as in Theorem \ref{thresolsing}%
. Then

\noindent(i) $Z_{\Theta}\left(  \boldsymbol{f},\boldsymbol{s}\right)  $ is
convergent and defines a holomorphic function in the region
\[
\sum_{j=1}^{m}N_{f_{j},i}\operatorname{Re}(s_{j})+v_{i}>0,\qquad\text{ for
}i\in T ;
\]

\noindent(ii) $Z_{\Theta}\left(  \boldsymbol{f},\boldsymbol{s}\right)  $
admits a meromorphic continuation to the whole $\mathbb{C}^{m}$, with poles
belonging to
\[
\bigcup_{i\in T}\text{ }\bigcup_{t\in\mathbb{N}}\left\{  \sum_{j=1}%
^{m}N_{f_{j},i}s_{j}+v_{i}+t=0\right\}  .
\]

\noindent(iii) In particular, in the case of one polynomial $f$, i.e., $m=1$,
the pair $\left(  N_{f,i},v_{i}\right)  $ is called the numerical datum of
$E_{i}$, and the set $\left\{  \left(  N_{f,i},v_{i}\right)  ;i\in T\right\}
$ is called the numerical data of the resolution $\sigma$. It is well known
that $\min_{i\in T}\frac{v_{i}}{N_{f,i}}$ does not depend on the choice of the
resolution $\sigma$ (this value is called the real log canonical threshold of
$f$). Furthermore, $Z_{\Theta}(f,s)$ is holomorphic in the half-space
$\operatorname{Re}(s)>-\min_{i\in T}\frac{v_{i}}{N_{f,i}}$, and the possible
poles of its meromorphic continuation belong to the set $\cup_{i\in T}\left(
-\frac{v_{i}+\mathbb{N}}{N_{f,i}}\right)  $.
\end{theorem}

The above theorem is a consequence of the work of many people: Gel'fand (I. M.
and S. I.), Bernstein, Atiyah, Igusa, Loeser, as far as we know. We will use
this result, as well as Lemma \ref{Lemma0}, along this article; the
formulation that we are giving here is the one we require. The formulation of
Lemma \ref{Lemma0} will be crucial for dealing with certain non-classical
local zeta functions that occur in Section \ref{Case I different}.

\section{Local Zeta Functions of Koba-Nielsen Type}

We consider $\mathbb{R}^{N-3}$ as an $\mathbb{R}$-analytic manifold, with
$N\geq4$, and use $\left\{  x_{2},\ldots,x_{N-2}\right\}  $ as a coordinate
system.\ In addition, we take
\begin{equation}
D_{N}:=\left\{  x\in\mathbb{R}^{N-3};%
{\displaystyle\prod\limits_{i=2}^{N-2}}
x_{i}\text{ }%
{\displaystyle\prod\limits_{i=2}^{N-2}}
\left(  1-x_{i}\right)  \text{ }%
{\displaystyle\prod\limits_{2\leq i<j\leq N-2}}
\left(  x_{i}-x_{j}\right)  =0\right\}  , \label{Eq_0}%
\end{equation}
and use $\prod\nolimits_{i=2}^{N-2}dx_{i}$ to denote the measure induced by
the top differential form $dx_{2}\wedge\ldots\wedge dx_{N-2}$.

\begin{definition}
\label{Definition_K:N_Zeta_func}A \textit{Koba-Nielsen} \textit{local zeta
function }is defined to be an integral of the form%
\begin{equation}
Z^{(N)}(\boldsymbol{s}):=%
{\displaystyle\int\limits_{\mathbb{R}^{N-3}\smallsetminus D_{N}}}
{\displaystyle\prod\limits_{j=2}^{N-2}}
\left\vert x_{j}\right\vert ^{s_{1j}}%
{\displaystyle\prod\limits_{j=2}^{N-2}}
\left\vert 1-x_{j}\right\vert ^{s_{\left(  N-1\right)  j}}%
{\displaystyle\prod\limits_{2\leq i<j\leq N-2}}
\left\vert x_{i}-x_{j}\right\vert ^{s_{ij}}%
{\displaystyle\prod\limits_{i=2}^{N-2}}
dx_{i}, \label{Eq_1}%
\end{equation}
where $\boldsymbol{s}=\left(  s_{ij}\right)  =\cup_{j=2}^{N-2}\left\{
s_{1j},s_{\left(  N-1\right)  j}\right\}  \cup\cup_{2\leq i<j\leq N-2}\left\{
s_{ij}\right\}  $ is a list consisting of $\mathbf{d}$ complex variables,
where
\begin{equation}
\mathbf{d}:\mathbf{=}\left\{
\begin{array}
[c]{lll}%
2(N-3)+\left(
\begin{array}
[c]{c}%
N-3\\
2
\end{array}
\right)  & \text{if} & N\geq5\\
&  & \\
2 & \text{if} & N=4
\end{array}
\right.  \text{ \ \ \ \ }=\frac{N\left(  N-3\right)  }{2}.
\label{Defenition_d}%
\end{equation}

\end{definition}

\noindent For later use in formulas, it will be convenient to put also
$s_{ij}=s_{ji}$ for any occurring $\left\{  i,j\right\}  $. For simplicity of
notation, we will put $\mathbb{R}^{N-3}$ instead of $\mathbb{R}^{N-3}%
\smallsetminus D_{N}$ in (\ref{Eq_1}), and similarly in other such integrals.

In order to regularize the integral (\ref{Eq_1}), we will use a partition of
$\mathbb{R}^{N-3}$ constructed using a smooth function $\chi:\mathbb{R}%
\rightarrow\mathbb{R}$ satisfying%
\begin{equation}
\chi\left(  x\right)  =\left\{
\begin{array}
[c]{ccc}%
1 & \text{if} & x\in\left[  -2,2\right] \\
&  & \\
0 & \text{if} & x\in\left(  -\infty,-2-\epsilon\right]  \cup\left[
2+\epsilon,+\infty\right)  ,
\end{array}
\right.  \label{Function_Chi}%
\end{equation}
for some fixed positive $\epsilon$ sufficiently small. The existence of such a
function is well-known, see e.g. \cite[Section 1.4]{Hormander}, \cite[Section
5.2]{Igusa}. Let us mention that the number $2$ was chosen in an arbitrary
form, the key point is that the interval $\left[  0,1\right]  $ is included in
the locus where $\chi\equiv1$.

\begin{definition}
\label{Definition_phi}For $I\subseteq\left\{  2,\ldots,N-2\right\}  $,
including the empty set, we set%
\begin{equation}
\varphi_{I}:\mathbb{R}^{N-3}\rightarrow\mathbb{R}:x\mapsto%
{\displaystyle\prod\limits_{i\in I}}
\chi\left(  x_{i}\right)
{\displaystyle\prod\limits_{i\notin I}}
\left(  1-\chi\left(  x_{i}\right)  \right)  , \label{Function_phi}%
\end{equation}
with the convention that $%
{\textstyle\prod\nolimits_{i\in\emptyset}}
\cdot\equiv1$.
\end{definition}

Then $\varphi_{I}\in C^{\infty}\left(  \mathbb{R}^{N-3}\right)  $ and
$\sum_{I}\varphi_{I}\left(  x\right)  \equiv1$, for $x\in\mathbb{R}^{N-3}$. By
using this partition of the unity, we have
\begin{equation}
Z^{(N)}(\boldsymbol{s})=\sum_{I}Z_{I}^{(N)}(\boldsymbol{s}) \label{Eq_2A}%
\end{equation}
with
\begin{equation}
Z_{I}^{(N)}(\boldsymbol{s}):=%
{\displaystyle\int\limits_{\mathbb{R}^{N-3}}}
\varphi_{I}\left(  x\right)
{\displaystyle\prod\limits_{j=2}^{N-2}}
\left\vert x_{j}\right\vert ^{s_{1j}}%
{\displaystyle\prod\limits_{j=2}^{N-2}}
\left\vert 1-x_{j}\right\vert ^{s_{\left(  N-1\right)  j}}%
{\displaystyle\prod\limits_{2\leq i<j\leq N-2}}
\left\vert x_{i}-x_{j}\right\vert ^{s_{ij}}\text{ }%
{\displaystyle\prod\limits_{i=2}^{N-2}}
dx_{i}. \label{Eq_3}%
\end{equation}
In the case $I=\left\{  2,\ldots,N-2\right\}  $, $Z_{I}^{(N)}(\boldsymbol{s})$
is a classical multivariate Igusa local zeta function (since then $\varphi
_{I}\left(  x\right)  $ has compact support). These integrals are holomorphic
functions in a region including $\operatorname{Re}\left(  s_{ij}\right)  >0$
for all $ij$, and they admit meromorphic continuations to the whole
$\mathbb{C}^{\mathbf{d}}$, see Theorem \ref{thm: num data and poles}.

In the case $I\neq\left\{  2,\ldots,N-2\right\}  $, by changing variables in
(\ref{Eq_3}) as $x_{i}\rightarrow\frac{1}{x_{i}}$ for $i\not \in I$, and
$x_{i}\rightarrow x_{i}$ for $i\in I$, we have $%
{\textstyle\prod\nolimits_{i=2}^{N-2}}
dx_{i}\rightarrow%
{\textstyle\prod\nolimits_{i\not \in I}}
\frac{1}{\left\vert x_{i}\right\vert ^{2}}$ $%
{\textstyle\prod\nolimits_{i=2}^{N-2}}
dx_{i}$, and by setting $\widetilde{\chi}\left(  x_{i}\right)  :=1-\chi\left(
\frac{1}{x_{i}}\right)  $ for $i\not \in I$, i.e.,%
\[
\widetilde{\chi}\left(  x_{i}\right)  =\left\{
\begin{array}
[c]{ccc}%
1 & \text{if} & \left\vert x_{i}\right\vert \leq\frac{1}{2+\epsilon}\\
&  & \\
0 & \text{if} & \left\vert x_{i}\right\vert \geq\frac{1}{2},
\end{array}
\right.  \text{,}%
\]
we have that supp $\widetilde{\chi} \subseteq\left[  -\frac{1}{2},\frac{1}%
{2}\right]  $ and $\widetilde{\chi}\in C^{\infty}\left(  \mathbb{R}\right)  $.
Now setting $\widetilde{\varphi}_{I}\left(  x\right)  :=\prod
\nolimits_{i\not \in I}\widetilde{\chi}\left(  x_{i}\right)  $ $\prod
\nolimits_{i\in I}\chi\left(  x_{i}\right)  $, and
\begin{align*}
F_{I}\left(  x,\boldsymbol{s}\right)   &  :=\prod\limits_{j\in I}\left\vert
x_{j}\right\vert ^{s_{1j}}\text{ }\prod\limits_{j=2}^{N-2}\left\vert
1-x_{j}\right\vert ^{s_{\left(  N-1\right)  j}}\prod\limits_{\substack{2\leq
i<j\leq N-2\\i,\text{ }j\in I}}\left\vert x_{i}-x_{j}\right\vert ^{s_{ij}%
}\text{ }\times\\
&  \prod\limits_{\substack{2\leq i<j\leq N-2\\i,\text{ }j\not \in
I}}\left\vert x_{i}-x_{j}\right\vert ^{s_{ij}}\text{ }\prod
\limits_{\substack{2\leq i<j\leq N-2\\i\not \in I,\text{ }j\in I}}\left\vert
1-x_{i}x_{j}\right\vert ^{s_{ij}}\prod\limits_{\substack{2\leq i<j\leq
N-2\\i\in I,\text{ }j\not \in I}}\left\vert 1-x_{i}x_{j}\right\vert ^{s_{ij}},
\end{align*}
we have
\begin{equation}
Z_{I}^{(N)}(\boldsymbol{s})=%
{\displaystyle\int\limits_{\mathbb{R}^{N-3}\smallsetminus D_{I}}}
\text{ }\frac{\widetilde{\varphi}_{I}\left(  x\right)  F_{I}\left(
x,\boldsymbol{s}\right)  }{\prod\limits_{i\not \in I}\left\vert x_{i}%
\right\vert ^{s_{1i}+s_{\left(  N-1\right)  i}+\sum_{\substack{2\leq j\leq
N-2\\j\neq i}}s_{ij}+2}\text{ }}%
{\displaystyle\prod\limits_{i=2}^{N-2}}
dx_{i}, \label{Eq 4}%
\end{equation}
where $D_{I}$ is the divisor defined by the polynomial
\begin{multline*}
\prod\limits_{i=2}^{N-2}x_{i}\text{ }\prod\limits_{i=2}^{N-2}\left(
1-x_{i}\right)  \text{ }\prod\limits_{\substack{2\leq i<j\leq N-2\\i,\text{
}j\in I}}\left(  x_{i}-x_{j}\right)  \prod\limits_{\substack{2\leq i<j\leq
N-2\\i,\text{ }j\notin I}}\left(  x_{i}-x_{j}\right)  \times\\
\text{ }\prod\limits_{\substack{2\leq i<j\leq N-2\\i\notin I,\text{ }j\in
I}}\left(  1-x_{i}x_{j}\right)  \text{ }\prod\limits_{\substack{2\leq i<j\leq
N-2\\\text{ }i\in I,\text{ }j\notin I}}\left(  1-x_{i}x_{j}\right)  .
\end{multline*}
From the expression in (\ref{Eq 4}), it is not clear at all whether integrals
of type $Z_{I}^{(N)}(\boldsymbol{s})$, with $I\neq\left\{  2,\ldots
,N-2\right\}  $, converge for some value of $\boldsymbol{s}$. These integrals
are not classical multivariate local zeta functions, and Theorem
\ref{thm: num data and poles} does not apply to them. We will show that they
define holomorphic functions on some nonempty open in $\mathbb{C}^{\mathbf{d}%
}$, and admit meromorphic continuations to the whole $\mathbb{C}^{\mathbf{d}}%
$. To establish this result we will use Lemma \ref{Lemma0} and embedded
resolution of singularities. The technical statement is as follows.

\begin{lemma}
\label{Lemma1} For any $I \subseteq\left\{  2,\ldots,N-2\right\}  $, the
function $Z_{I}^{(N)}(\boldsymbol{s})$ is holomorphic in $\boldsymbol{s}$ on
the solution set $\mathcal{H}(I)$ of a system of inequalities of the form%
\begin{equation}
\mathcal{H}(I):=\left\{  s_{ij}\in\mathbb{C}^{\mathbf{d}};\text{ }\sum_{ij\in
M\left(  I\right)  }N_{ij,k}\left(  I\right)  \operatorname{Re}\left(
s_{ij}\right)  +\gamma_{k}\left(  I\right)  >0,\text{ for }k\in T(I)\right\}
\neq\emptyset, \label{Eq_7}%
\end{equation}
where $N_{ij,k}\left(  I\right)  ,\gamma_{k}\left(  I\right)  \in\mathbb{Z}$,
and $M(I)$, $T(I)$ are finite sets. More precisely, for each $k$, either all
numbers $N_{ij,k}\left(  I\right)  $ are equal to $0$ or $1$ and $\gamma
_{k}\left(  I\right)  > 0$, or all numbers $N_{ij,k}\left(  I\right)  $ are
equal to $0$ or $-1$ and $\gamma_{k}\left(  I\right)  < 0$.

In addition, $Z_{I}^{(N)}(\boldsymbol{s})$ admits an analytic continuation to
the whole $\mathbb{C}^{\mathbf{d}}$, as a meromorphic function with poles
belonging to
\begin{equation}
\mathcal{P}(I):=\bigcup\limits_{t\in\mathbb{N}}\bigcup\limits_{k\in
T(I)}\left\{  s_{ij}\in\mathbb{C}^{\mathbf{d}};\sum_{ij\in M\left(  I\right)
}N_{ij,k}(I)s_{ij}+\gamma_{k}(I)+t=0\right\}  . \label{Eq_7A}%
\end{equation}

\end{lemma}

\begin{proof}
By applying Hironaka's Theorem \ref{thresolsing} to the divisors $D_{N}$ or
$D_{I}$, and by using a suitable partition of the unity, each integral
$Z_{I}^{(N)}(\boldsymbol{s})$ becomes a finite sum of monomial type integrals.
The important statement here is that $\mathcal{H}(I)\neq\emptyset$ if
$I\neq\left\{  2,\ldots,N-2\right\}  $; this is shown later in Propositions
\ref{Prop1}, \ref{Prop2}, together with the specifications concerning the
numbers $N_{ij,k}\left(  I\right)  $ and $\gamma_{k}\left(  I\right)  $. Then
the meromorphic continuation of the integrals follows from Lemma \ref{Lemma0}.
\end{proof}

Then, by formula (\ref{Eq_2A}), the integral $Z^{(N)}(\boldsymbol{s})$ will be
a finite sum of functions $Z_{I}^{(N)}(\boldsymbol{s})$, holomorphic on the
domain $\mathcal{H}(I)$ in $\mathbb{C}^{\mathbf{d}}$, that however depends on
$I$. Hence, the convergence and the analytic continuation of the integral
$Z^{(N)}(\boldsymbol{s})$ is not a direct consequence of the existence of
meromorphic continuations for the integrals $Z_{I}^{(N)}(\boldsymbol{s})$. We
will show that all the integrals $Z_{I}^{(N)}(\boldsymbol{s})$ are holomorphic
in a common domain, and then formula (\ref{Eq_2A}) allows us to construct a
meromorphic continuation of $Z^{(N)}(\boldsymbol{s})$.

More precisely, in order to show that $\cap_{I}\mathcal{H}(I)$ contains a
non-empty open subset of $\mathbb{C}^{\mathbf{d}}$, we will take
$\operatorname{Re}\left(  s_{ij}\right)  =\operatorname{Re}(s)$ for any $ij$
and any $I$, and show that the solution set (of the system of inequalities
obtained in this way) contains a non-empty open subset of $\mathbb{C}$. This
fact will be established by studying possible poles of functions of the form%
\[
Z_{I}^{(N)}(s):=Z_{I}^{(N)}(\boldsymbol{s})\mid_{s_{ij}=s}%
\]
and proving that $Z^{(N)}(s):=Z^{(N)}(\boldsymbol{s})\mid_{s_{ij}=s}$ is a
holomorphic function in the region
\[
-\frac{2}{N-2}<\operatorname{Re}(s)<-\frac{2}{N},
\]
see again Propositions \ref{Prop1}, \ref{Prop2}.
A somewhat more elaborate argument yields the following more precise result.

\begin{theorem}
\label{TheoremB}The Koba-Nielsen local zeta function $Z^{(N)}(\boldsymbol{s})$
is a holomorphic function in the solution set $\cap_{I}\mathcal{H}(I)$, see
(\ref{Eq_7}), in $\mathbb{C}^{\mathbf{d}}$, which contains the set
\begin{equation}
-\frac{2}{N-2}<\operatorname{Re}(s_{ij})<-\frac{2}{N}\qquad\text{ for all
}ij\text{.} \label{EQ_set}%
\end{equation}
Furthermore, it has a meromorphic continuation, denoted again as
$Z^{(N)}(\boldsymbol{s})$, to the whole $\mathbb{C}^{\mathbf{d}}$, with poles
belonging to $\cup_{I}\mathcal{P}(I)$, see (\ref{Eq_7A}).
\end{theorem}

The meromorphic continuation of $Z^{(N)}(\boldsymbol{s})$ does not depend on
the choice of the function $\chi$. Suppose that we pick another smooth
function $\chi_{0}$, with compact support, such that $\left[  -1,1\right]
\subset$ supp $\chi_{0}$ and $\chi_{0}\mid_{\left[  -1,1\right]  }\equiv1$.
Then $Z^{(N)}(\boldsymbol{s})$ has a meromorphic continuation, denoted now as
$Z_{0}^{(N)}(\boldsymbol{s})$, to the whole $\mathbb{C}^{\mathbf{d}}$ minus a
countable number of hyperplanes. Both $Z^{(N)}(\boldsymbol{s})$ and
$Z_{0}^{(N)}(\boldsymbol{s})$ are holomorphic in $\mathbb{C}^{\mathbf{d}}$
minus a countable number of hyperplanes and coincide in an open set of
$\mathbb{C}^{\mathbf{d}}$, where both functions are holomorphic. Consequently,
by the analytic continuation principle, $Z^{(N)}(\boldsymbol{s})=Z_{0}%
^{(N)}(\boldsymbol{s})$ in $\mathbb{C}^{\mathbf{d}}$ minus a countable number
of hyperplanes.

\section{\label{Section_Road_map}Road Map of the Proof}

Since the $D_{I}$ are so-called hyperplane arrangements (their irreducible
components are all hyperplanes), there is a well-known and straightforward way
to construct an embedded resolution for them. In fact, we only have to modify
the locus where the defining equation of $D_{I}$ is \emph{not} locally
monomial and that moreover is contained in the support of $\varphi_{I}$ or
$\widetilde{\varphi}_{I}$. The standard algorithm is to blow up consecutively
in relevant centres of increasing dimension contained in that locus, until the
total inverse image of $D_{I}$ becomes a normal crossings divisor, i.e., its
defining equation becomes locally monomial. Then we conclude by repeated
applications of Lemma \ref{Lemma0}.

For readers who are not familiar with these notions, we will treat explicitly
the first blow-ups of such a resolution, presented as explicit change of
variables operations, simplifying the original integrals $Z_{I}^{(N)}%
(\boldsymbol{s})$.

First, in order to get a feeling for the method, we present the easy case
$N=4$ and especially the case $N=5$, where we explain in detail the set
$\mathcal{H}(I)$ of Lemma \ref{Lemma1}.

\subsection{Example\label{Example_N_4}}

Fix $N=4$. Then $Z^{(4)}(s)=Z_{\{2\}}^{(4)}(s)+Z_{\emptyset}^{(4)}(s)$. We
have that
\[
Z_{\{2\}}^{(4)}(\boldsymbol{s})=\int_{\mathbb{R}}\chi(x_{2})|x_{2}|^{s_{12}%
}|1-x_{2}|^{s_{32}}dx_{2},
\]
which converges and is analytic when $\operatorname{Re}(s_{12})+1>0$ and
$\operatorname{Re}(s_{32})+1>0$, and
\[
Z_{\emptyset}^{(4)}(\boldsymbol{s})=\int_{\mathbb{R}}\widetilde{\chi}%
(x_{2})|x_{2}|^{-s_{12}-s_{32}-2}|1-x_{2}|^{s_{32}}dx_{2},
\]
which converges and is analytic when $-\operatorname{Re}(s_{12}%
)-\operatorname{Re}(s_{32})-1>0$ and $\operatorname{Re}(s_{32})+1>0$. We
conclude by Lemma \ref{Lemma0} that $Z^{(4)}(s)$ converges and is analytic in
the region
\[
\operatorname{Re}(s_{12})>-1,\quad\operatorname{Re}(s_{32})>-1,\quad
\operatorname{Re}(s_{12})+\operatorname{Re}(s_{32})<-1.
\]

\subsection{Example\label{Example_N_5}}

Fix $N=5$. Then $Z^{(5)}(\boldsymbol{s})=Z_{\{2,3\}}^{(5)}(\boldsymbol{s}%
)+Z_{\{3\}}^{(5)}(\boldsymbol{s})+Z_{\{2\}}^{(5)}(\boldsymbol{s}%
)+Z_{\emptyset}^{(5)}(\boldsymbol{s})$. We start with
\begin{equation}
Z_{\{2,3\}}^{(5)}(\boldsymbol{s})=\int_{\mathbb{R}^{2}}\chi(x_{2})\chi
(x_{3})|x_{2}|^{s_{12}}|x_{3}|^{s_{13}}|1-x_{2}|^{s_{42}}|1-x_{3}|^{s_{43}%
}|x_{2}-x_{3}|^{s_{23}}dx_{2}dx_{3}. \label{Z23}%
\end{equation}
The arrangement $D_{5}$, given by $x_{2}x_{3}(1-x_{2})(1-x_{3})(x_{2}%
-x_{3})=0$, is not locally monomial only at the points $(0,0)$ and $(1,1)$. We
pick a partition of the unity, $\sum_{i=0}^{2}$ $\Omega_{i}(x_{2},x_{3})=1$,
where $\Omega_{0}$ and $\Omega_{1}$ are smooth functions with support in a
small neighborhood of $(0,0)$ and $(1,1)$, respectively. This reduces the
computation to neighborhoods of each of these two points. That is, we can
write $Z_{\{2,3\}}^{(5)}(s)=\sum_{j=0}^{2}Z_{\Omega_{j}}^{(5)}(s)$, where the
$Z_{\Omega_{j}}^{(5)}(s)$ are integrals as in (\ref{Z23}), replacing
$\chi(x_{2})\chi(x_{3})$ by $\Omega_{j}(x_{2},x_{3})$, and the relevant
integrals to compute are $Z_{\Omega_{0}}^{(5)}(s)$ and $Z_{\Omega_{1}}%
^{(5)}(s)$.

Around $(0,0)$, the factor $|1-x_{2}|^{s_{42}}|1-x_{3}|^{s_{43}}$ can be
neglected from the point of view of convergence and holomorphy, hence here we
only need an embedded resolution of $x_{2}x_{3}(x_{2}-x_{3})=0$, which is
obtained by a blow-up at the origin. This just means the two changes of
variables
\begin{align*}
\sigma_{0}:\mathbb{R}^{2}\rightarrow\mathbb{R}^{2}:  &  \ u_{2}\mapsto
x_{2}=u_{2} & \quad\text{and}\quad\sigma_{0}^{\prime}:\mathbb{R}%
^{2}\rightarrow\mathbb{R}^{2}:  &  \ u_{2}\mapsto x_{2}=u_{2}u_{3}\\
&  \ u_{3}\mapsto x_{3}=u_{2}u_{3} &  &  \ u_{3}\mapsto x_{3}=u_{3}.
\end{align*}
By symmetry it is enough to consider the first one, yielding as contribution
to $Z_{\Omega_{0}}^{(5)}(s)$ the integral
\[
{\displaystyle\int\limits_{\mathbb{R}^{2}}}\left(  \Omega_{0}\circ\sigma
_{0}\right)  \left(  u_{2},u_{3}\right)  \left\vert u_{2}\right\vert
^{s_{12}+s_{13}+s_{23}+1}|u_{3}|^{s_{13}}|1-u_{3}|^{s_{23}}g(u,\boldsymbol{s}%
)du_{2}du_{3},
\]
where $g(u,s)$ is invertible on the support of $\Omega_{0}\circ\sigma_{0}$
(and can thus be neglected from the point of view of convergence and
holomorphy). Hence, we are in the locally monomial setting, and Lemma
\ref{Lemma0} yields the convergence conditions
\begin{equation}
\operatorname{Re}(s_{12})+\operatorname{Re}(s_{13})+\operatorname{Re}%
(s_{23})+2>0,\quad\operatorname{Re}(s_{23})+1>0,\quad\operatorname{Re}%
(s_{13})+1>0. \label{ex1}%
\end{equation}
The other chart of the blow-up, i.e., the change of variables $\sigma
_{0}^{\prime}$, yields the same first and second condition and also
\begin{equation}
\operatorname{Re}(s_{12})+1>0. \label{ex2}%
\end{equation}
Completely similarly, for the convergence of $Z_{\Omega_{1}}^{(5)}(s)$, we
need also the new conditions
\begin{equation}
\operatorname{Re}(s_{42})+\operatorname{Re}(s_{43})+\operatorname{Re}%
(s_{23})+2>0,\quad\operatorname{Re}(s_{42})+1>0,\quad\operatorname{Re}%
(s_{43})+1>0. \label{ex3}%
\end{equation}
Note that the conditions coming from the locally monomial integral
$Z_{\Omega_{2}}^{(5)}(s)$ are already included now. Next,
\begin{equation}
\label{first infinity}Z_{\{3\}}^{(5)}(\boldsymbol{s})=\int_{\mathbb{R}^{2}%
}\widetilde{\chi}(x_{2})\chi(x_{3})|x_{2}|^{-s_{12}-s_{42}-s_{23}-2}%
|x_{3}|^{s_{13}}|1-x_{2}|^{s_{42}}|1-x_{3}|^{s_{43}}|1-x_{2}x_{3}|^{s_{23}%
}dx_{2}dx_{3}.
\end{equation}
Since $x_{2}x_{3}(1-x_{2})(1-x_{3})(1-x_{2}x_{3})$ is locally monomial in the
support of $\widetilde{\chi}(x_{2})\chi(x_{3})$, the only new condition that
arises is
\begin{equation}
-\operatorname{Re}(s_{12})-\operatorname{Re}(s_{42})-\operatorname{Re}%
(s_{23})-1>0. \label{ex4}%
\end{equation}
Completely analogously, $Z_{\{2\}}^{(5)}(s)$ induces the extra condition
\begin{equation}
-\operatorname{Re}(s_{13})-\operatorname{Re}(s_{43})-\operatorname{Re}%
(s_{23})-1>0. \label{ex5}%
\end{equation}
Finally, we must consider
\[
Z_{\emptyset}^{(5)}(\boldsymbol{s})=\int_{\mathbb{R}^{2}}\widetilde{\chi
}(x_{2})\widetilde{\chi}(x_{3})|x_{2}|^{-s_{12}-s_{42}-s_{23}-2}%
|x_{3}|^{-s_{13}-s_{43}-s_{23}-2}|1-x_{2}|^{s_{42}}|1-x_{3}|^{s_{43}}%
|x_{2}-x_{3}|^{s_{23}}dx_{2}dx_{3}.
\]
Geometrically, we have here the same arrangement as for $Z_{\left\{
2,3\right\}  }^{(5)}(s)$; the differences\textrm{ }are the powers of $|x_{2}|$
and $|x_{3}|$ and the function $\widetilde{\chi}(x_{2})\widetilde{\chi}%
(x_{3})$, that does not contain $(1,1)$ in its support. Hence, the only new
condition will arise from the blow-up at the origin, namely
\[
(-\operatorname{Re}(s_{12})-\operatorname{Re}(s_{42})-\operatorname{Re}%
(s_{23})-2)+(-\operatorname{Re}(s_{13})-\operatorname{Re}(s_{43}%
)-\operatorname{Re}(s_{23})-2)+\operatorname{Re}(s_{23})+2>0,
\]
which simplifies to
\begin{equation}
-\operatorname{Re}(s_{12})-\operatorname{Re}(s_{42})-\operatorname{Re}%
(s_{23})-\operatorname{Re}(s_{13})-\operatorname{Re}(s_{43})-2>0. \label{ex6}%
\end{equation}
The conclusion is that $Z^{(5)}(s)$ converges and is analytic in the region
defined by conditions (\ref{ex1})-(\ref{ex6}), that is, for the five indices
$ij$,
\begin{equation}
\begin{aligned} & \text{all } \operatorname{Re}(s_{ij}) > -1 ,\\ & \operatorname{Re}(s_{12})+ \operatorname{Re}(s_{13})+\operatorname{Re}(s_{23}) > -2, \quad \operatorname{Re}(s_{42})+\operatorname{Re}(s_{43})+\operatorname{Re}(s_{23}) > -2 , \\ & \operatorname{Re}(s_{12})+\operatorname{Re}(s_{42})+\operatorname{Re}(s_{23})< -1 , \quad \operatorname{Re}(s_{13})+\operatorname{Re}(s_{43})+\operatorname{Re}(s_{23}) < -1 ,\\ & \sum_{ij} \operatorname{Re}(s_{ij}) < -2 . \end{aligned} \label{ex7}%
\end{equation}
It is immediate to verify that the simpler domain
\begin{equation}
-\frac{2}{3}<\operatorname{Re}(s_{ij})<-\frac{2}{5}\qquad\text{ for all }ij
\label{ex8}%
\end{equation}
is contained in this region. Then, in particular, $Z^{(5)}(s)$ is analytic in
the interval $-\frac{2}{3}<\operatorname{Re}(s)<-\frac{2}{5}$.

\bigskip We now proceed in general, assuming implicitly that $N\geq6$. The
regions generalizing (\ref{ex7}) can be determined explicitly for any given
$N$, but require long lists of inequalities. They will be stated in
Proposition \ref{exact convergence domain}. To keep the exposition more
accessible, we will concentrate on the domain generalizing (\ref{ex8}) and its
specialization to $Z^{(N)}(s)$.

\subsection{Case $I=\left\{  2,\ldots,N-2\right\}  $}

Recall that then $Z_{I}^{(N)}(\boldsymbol{s})$ and $Z_{I}^{(N)}(s)$ are
classical local zeta functions, the last one associated to the polynomial
\[
f_{N}\left(  x\right)  =\prod_{i=2}^{N-2}x_{i}\prod_{i=2}^{N-2}(1-x_{i}%
)\prod_{2\leq i<j\leq N-2}(x_{i}-x_{j}).
\]
We note first that, for convergence of $Z_{I}^{(N)}(\boldsymbol{s})$, at least
the conditions $\operatorname{Re}(s_{ij})+1>0$, for all $ij$, are needed,
since they will certainly appear as powers of (absolute values of) variables
in some monomial integral. Geometrically, these conditions are induced by the
strict transforms $E_{i}$ in an embedded resolution of the components of
$D_{N}$. ( For convergence of $Z_{I}^{(N)}(s)$, they all induce the same
condition $\operatorname{Re}(s)+1>0$.)

Next, concerning $Z_{I}^{(N)}(s)$, we have by Theorem
\ref{thm: num data and poles}(iii) that this function is holomorphic in the
half-space
\[
\operatorname{Re}(s)>-\min_{i\in T}\frac{v_{i}}{N_{f_{N},i}},
\]
where $\left\{  \left(  N_{f_{N},i},v_{i}\right)  ;i\in T\right\}  $ are the
numerical data of an embedded resolution $\sigma$ of $D_{N}$. We will explain
how to construct such a resolution and obtain that this minimum value is
$\frac{2}{N-2}$. Working out more details yields a concrete region where
$Z_{I}^{(N)}(\boldsymbol{s})$ is holomorphic.

Note that the locus of $D_{N}$ where it is not a normal crossings divisor,
i.e., not locally monomial as in Lemma \ref{Lemma0}, consists of the points
with at least two coordinates equal to $0$, at least two coordinates equal to
$1$, or at least three equal coordinates. We will blow up consecutively in the
relevant centres of increasing dimension contained in that locus, until the
total inverse image of $D_{N}$ becomes a normal crossings divisor.


\begin{proposition}
\label{Prop1} We take $I=\left\{  2,\ldots,N-2\right\}  $. Then $Z_{I}%
^{(N)}(\boldsymbol{s})$ is convergent and holomorphic in a region containing
$\operatorname{Re}(s_{ij}) >-\frac{2}{N-2} $ for all $ij$. In particular, the
integral $Z_{I}^{(N)}(s)$ is holomorphic in the half-plane $\operatorname{Re}
(s)>-\frac{2}{N-2} $.
\end{proposition}

\begin{proof}
We recall that
\[
Z_{I}^{(N)}(\boldsymbol{s}):={\displaystyle\int\limits_{\mathbb{R}^{N-3}}}
\varphi_{I}\left(  x\right)  F_{I}(x,\boldsymbol{s}) {\displaystyle\prod
\limits_{i=2}^{N-2}} dx_{i},
\]
where
\[
F_{I}(x,\boldsymbol{s}) = {\displaystyle\prod\limits_{j=2}^{N-2}} \left\vert
x_{j}\right\vert ^{s_{1j}}{\displaystyle\prod\limits_{j=2}^{N-2}} \left\vert
1-x_{j}\right\vert ^{s_{\left(  N-1\right)  j}}{\displaystyle\prod
\limits_{2\leq i<j\leq N-2}} \left\vert x_{i}-x_{j}\right\vert ^{s_{ij}}.
\]
First, we consider an adequate partition of the unity subordinate to the
compact set $\operatorname{supp}\varphi_{I}$. Let $P$ be the set of $2^{N-3}$
points $p$ in $\mathbb{R}^{N-3}$ with all coordinates equal to $0$ or $1$. We
take smooth functions $\Omega_{p}$, $p\in P$, such that each $\Omega_{p}$ is
supported in a neighborhood of $p$ that is disjoint from some neighborhood of
any other point of $P$, and such that $\varphi_{I}(x)=\sum_{p\in P}\Omega
_{p}\left(  x\right)  $ for $x\in\operatorname{supp}\varphi_{I}$. Then
\[
Z_{I}^{(N)}(\boldsymbol{s})=\sum_{p\in P}Z_{\Omega_{p}}^{(N)}(\boldsymbol{s}%
),\quad\text{where}\quad Z_{\Omega_{p}}^{(N)}(\boldsymbol{s}):={\int
\limits_{\mathbb{R}^{N-3}}}\Omega_{p}\left(  x\right)  F_{I}(x,\boldsymbol{s})
{\displaystyle\prod\limits_{i=2}^{N-2}}dx_{i}.
\]

\textbf{(1)} We start by improving the situation around the origin ($p=0$). We
remark that the factors $|1-x_{j}|^{s_{(N-1)j}}$ in the integrand of
$Z_{\Omega_{0}}^{(N)}(s)$ are invertible on the support of $\Omega_{0}$, and
can be neglected from the point of view of convergence and holomorphy of
$Z_{\Omega_{0}}^{(N)}(\boldsymbol{s})$. Hence, we want in this stage an
embedded resolution of the divisor $D_{0}$ given by the zero locus of
\[
\prod_{i=2}^{N-2}x_{i}\prod_{2\leq i<j\leq N-2}(x_{i}-x_{j}).
\]

\textbf{(1.1)} The blow-up at the origin of $\mathbb{R}^{N-3}$ involves $N-3$
changes of variables of type $x_{i_{0}}=u_{i_{0}}$, $x_{i}=u_{i_{0}}u_{i}$ for
$i\in\left\{  2,\ldots,N-2\right\}  \setminus\left\{  i_{0}\right\}  $ and
some fixed $i_{0}\in\left\{  2,\ldots,N-2\right\}  $. This change of variables
defines
\begin{equation}
\sigma_{0}:\mathbb{R}^{N-3}\rightarrow\mathbb{R}^{N-3}:u\mapsto x.
\label{chart_1}%
\end{equation}
Without loss of generality, we may assume that $i_{0}=2$. Then $F_{I}%
(x,\boldsymbol{s}) \circ\sigma_{0} $ equals
\[
|u_{2}|^{\sum_{j=2}^{N-2} s_{1j} + \sum_{2\leq i<j\leq N-2} s_{ij} }
{\prod\limits_{j=3}^{N-2}}\left\vert u_{j}\right\vert ^{s_{1j}}{\prod
\limits_{j=3}^{N-2}}\left\vert 1-u_{j}\right\vert ^{s_{(N-1)j}} {\prod
\limits_{3\leq i<j\leq N-2}}\left\vert u_{i}-u_{j}\right\vert ^{s_{ij}%
}g(u,\boldsymbol{s}),
\]
where%
\[
g(u,\boldsymbol{s}):=\left\vert 1-u_{2}\right\vert ^{s_{(N-1)2}}%
{\displaystyle\prod\limits_{j=3}^{N-2}} \left\vert 1-u_{2}u_{j}\right\vert
^{s_{(N-1)j}},
\]
and
\[
\sigma_{0}^{\ast}{\displaystyle\prod\limits_{i=2}^{N-2}} dx_{i}=\left\vert
u_{2}\right\vert ^{N-4}{\displaystyle\prod\limits_{i=2}^{N-2}} du_{i}.
\]
Hence, the contribution to $Z_{\Omega_{0}}^{(N)}(\boldsymbol{s})$ in the chart
(\ref{chart_1}), with $i_{0}=2$, takes the form
\[
\begin{aligned}
{\displaystyle\int\limits_{\mathbb{R}^{N-3}}}
\left(  \Omega_{0}\circ\sigma_{0}\right)(u)  &
|u_2|^{\sum_{j=2}^{N-2} s_{1j} + \sum_{2\leq i<j\leq N-2}  s_{ij} +N-4 }
{\prod\limits_{j=3}^{N-2}}\left\vert
u_{j}\right\vert ^{s_{1j}}{\prod\limits_{j=3}^{N-2}}\left\vert 1-u_{j}\right\vert
^{s_{(N-1)j}}  \times \\
&{\prod\limits_{3\leq i<j\leq N-2}}\left\vert u_{i}-u_{j}\right\vert
^{s_{ij}}g(u,\boldsymbol{s})  {\displaystyle\prod\limits_{i=2}^{N-2}} du_{i},
\label{u_2}%
\end{aligned}
\]
where the factor $g(u,\boldsymbol{s})$ is invertible on the support of
$\Omega_{0}\circ\sigma_{0}$, and can be neglected from the point of view of
convergence and holomorphy. The first new condition for convergence is thus
\begin{equation}
\label{firstcondition}\sum_{j=2}^{N-2} \operatorname{Re}( s_{1j}) +
\sum_{2\leq i<j\leq N-2} \operatorname{Re}( s_{ij}) +N-3 > 0.
\end{equation}
Simplifying all $s_{ij}$ to $s$, the contribution to $Z_{\Omega_{0}}^{(N)}(s)$
in this chart is
\[
{\displaystyle\int\limits_{\mathbb{R}^{N-3}}} \left(  \Omega_{0}\circ
\sigma_{0}\right)  \left(  u\right)  \left\vert u_{2}\right\vert
^{\frac{\left(  N-2\right)  \left(  N-3\right)  }{2}s+N-4}{\displaystyle\prod
\limits_{i=3}^{N-2}} \left\vert u_{i}\right\vert ^{s}{\prod\limits_{i=3}%
^{N-2}}\left\vert 1-u_{i}\right\vert ^{s}{\displaystyle\prod\limits_{3\leq
i<j\leq N-2}} \left\vert u_{i}-u_{j}\right\vert ^{s}g(u,s){\displaystyle\prod
\limits_{i=2}^{N-2}} du_{i}.
\]
Up to negligible factors as $g(u,\boldsymbol{s})$, further blowing-ups/changes
of variables, ultimately leading to monomial integrals as in Lemma
\ref{Lemma0}, will not affect the variable $u_{2}$ anymore.

The smooth hypersurface, given by $u_{2}=0$, corresponds to a submanifold
$E_{0}$ (as in Theorem \ref{thresolsing}) with numerical data
\begin{equation}
(N_{f_{N},0},v_{0})=\left(  \frac{\left(  N-2\right)  \left(  N-3\right)  }%
{2},N-3\right)  ,\qquad\text{satisfying}\qquad\frac{v_{0}}{N_{f_{N},0}}%
=\frac{2}{N-2}. \label{minimal quotient}%
\end{equation}
Important to note is that $N_{f_{N},0}$ is equal to the multiplicity of
$D_{N}$ at the origin, which, in the case of a hyperplane arrangement, is just
the number of hyperplanes containing the origin. Also, $v_{0}$ is equal to the
codimension of the origin in $\mathbb{R}^{N-3}$. This is a general fact:
\textit{for any submanifold }$E_{i}$\textit{ as in Theorem \ref{thresolsing},
created by a blow-up with centre }$Y$\textit{, we have that }$N_{f_{N},i}%
$\textit{ is equal to the multiplicity of }$D_{N}$\textit{ at (a generic point
of) }$Y$\textit{, being the number of hyperplanes containing }$Y$\textit{, and
that }$v_{i}$\textit{ is equal to the codimension of }$Y$\textit{ in the
ambient space.}

\smallskip\textbf{(1.2)} The next blow-ups, in centres intersecting $E_{0}$,
are at those centres of dimension 1 whose image by $\sigma_{0}$ contains the
origin. There are two such centres visible in the present chart. The first one
is $u_{3}=\ldots=u_{N-2}=0$ (this is the so-called strict transform of the
line $x_{3}=\ldots=x_{N-2}=0$). This blow-up consists of $N-4$ changes of
variables of type $u_{i_{1}}=w_{i_{1}}$, for some $i_{1}\in\left\{
3,\ldots,N-2\right\}  $ and $u_{i}=w_{i_{1}}w_{i}$ for $i\in\left\{
3,\ldots,N-2\right\}  $, $i\neq i_{1}$, and $u_{2}=w_{2}$. We pick $i_{1}=3$
(the other cases are treated in a similar way), and take thus $u_{2}=w_{2}$,
$u_{3}=w_{3}$, and $u_{i}=w_{3}w_{i}$ for $i\in\left\{  4,\ldots,N-2\right\}
$, defining the change of variables
\[
\sigma_{1}:\mathbb{R}^{N-3}\rightarrow\mathbb{R}^{N-3}:w\mapsto u.
\]
Then in this chart the contribution to $Z_{\Omega_{0}}^{(N)}(\boldsymbol{s})$
takes the form%
\[
\begin{aligned}
{\displaystyle\int\limits_{\mathbb{R}^{N-3}}}
(\Omega_{0}\circ\sigma_{0}\circ\sigma_{1})(w)  &  \left\vert w_{2}\right\vert
^{\sum_{j=2}^{N-2} s_{1j} + \sum_{2\leq i<j\leq N-2}  s_{ij}   +N-4}\left\vert
w_{3}\right\vert ^{  \sum_{j=3}^{N-2} s_{1j} + \sum_{3\leq i<j\leq N-2}  s_{ij}
+N-5}\times\\%
& {\displaystyle\prod\limits_{j=4}^{N-2}}
\left\vert w_{j}\right\vert ^{s_{1j}}%
{\displaystyle\prod\limits_{4\leq j\leq N-2}}
\left\vert 1-w_{j}\right\vert ^{s_{(N-1)j}}%
{\displaystyle\prod\limits_{4\leq i<j\leq N-2}}
\left\vert w_{i}-w_{j}\right\vert ^{s_{ij}}h(w,\boldsymbol{s})%
{\displaystyle\prod\limits_{i=2}^{N-2}}
dw_{i}, 
\end{aligned}
\]
where the factor $h(w,\boldsymbol{s})$ can be neglected from the point of view
of convergence and holomorphy. The next condition for convergence is thus
\begin{equation}
\sum_{j=3}^{N-2}\operatorname{Re}(s_{1j})+\sum_{3\leq i<j\leq N-2}%
\operatorname{Re}(s_{ij})+N-4>0. \label{secondcondition}%
\end{equation}
Simplifying, the contribution to $Z_{\Omega_{0}}^{(N)}(s)$ is
\[
\begin{aligned}{\displaystyle\int\limits_{\mathbb{R}^{N-3}}}
(\Omega_{0}\circ\sigma_{0}\circ\sigma_{1})(w)  &  \left\vert w_{2}\right\vert
^{\frac{\left(  N-2\right)  \left(  N-3\right)  }{2}s+N-4}\left\vert
w_{3}\right\vert ^{\frac{\left(  N-3\right)  \left(  N-4\right)  }{2}%
s+N-5}\times\\%
&
{\displaystyle\prod\limits_{i=4}^{N-2}}
\left\vert w_{i}\right\vert ^{s}%
{\displaystyle\prod\limits_{4\leq j\leq N-2}}
\left\vert 1-w_{i}\right\vert ^{s}%
{\displaystyle\prod\limits_{4\leq i<j\leq N-2}}
\left\vert w_{i}-w_{j}\right\vert ^{s}h(w,s)%
{\displaystyle\prod\limits_{i=2}^{N-2}}
dw_{i}. 
\end{aligned}
\]
The smooth hypersurface, given by $w_{3}=0$, corresponds to a submanifold
$E_{1}$ with numerical data
\[
(N_{f_{N},1},v_{1})=\left(  \frac{\left(  N-3\right)  \left(  N-4\right)  }%
{2},N-4\right)  ,\qquad\text{satisfying}\qquad\frac{v_{1}}{N_{f_{N},1}}%
=\frac{2}{N-3}.
\]

The second centre is $1=u_{3}=\ldots=u_{N-2}$ (the strict transform of the
line $x_{2}=x_{3}=\ldots=x_{N-2}$). After a change of variables $u^{\prime
}_{i}=u_{i}-1$ for $i=3,\ldots,N-2$, the calculation of this blow-up is the
same as for the first centre. It gives rise to a submanifold $E^{\prime}_{1}$
with the same numerical data $(N^{\prime}_{f_{N},1},v^{\prime}_{1}) =\left(
\frac{\left(  N-3\right)  \left(  N-4\right)  }{2}, N-4\right)  $, yielding
the same quotient $\frac{v^{\prime}_{1}}{N^{\prime}_{f_{N},1}} = \frac{2}%
{N-3}$. The associated new condition for convergence for $Z_{\Omega_{0}}%
^{(N)}(\boldsymbol{s})$ is
\begin{equation}
\label{thirdcondition}\sum_{2\leq i<j\leq N-2} \operatorname{Re}( s_{ij}) +N-4
> 0.
\end{equation}

\textbf{(1.3)} We continue this way, blowing up in centres of increasing
dimension, ending with blow-ups in centres of dimension $N-5$ of two possible
types, for instance corresponding to $x_{N-3}=x_{N-2}=0$ and $x_{N-4}%
=x_{N-3}=x_{N-2}$, respectively, yielding submanifolds with numerical data
$(3,2)$. The corresponding conditions for convergence for $Z_{\Omega_{0}%
}^{(N)}(\boldsymbol{s})$ are
\begin{equation}
\label{lastcondition1}\operatorname{Re}(s_{1(N-3)}) + \operatorname{Re}%
(s_{1(N-2)}) + \operatorname{Re}(s_{(N-3)(N-2)}) +2 > 0
\end{equation}
and
\begin{equation}
\label{lastcondition2}\operatorname{Re}(s_{(N-4)(N-3)}) + \operatorname{Re}%
(s_{(N-4)(N-2)}) + \operatorname{Re}(s_{(N-3)(N-2)}) +2 > 0 .
\end{equation}
Note that, up to now, the smallest quotient of numerical data that we obtained
is indeed $\frac{2}{N-2}$.

\smallskip\textbf{(2)} All other points $p=\left(  p_{2},\ldots,p_{N-2}%
\right)  \in P$, that are needed as centres of blow-ups, have at least one
coordinate equal to $1$ (and still at least two coordinates equal to $0$ or at
least two coordinates equal to $1$), say $p_{i}=1$ for $i\in J\neq\emptyset$
and $p_{i}=0$ for $i\notin J$. For simplicity, we switch to the coordinate
system $y$, given by $y_{i}=x_{i}-1$ for $i\in J$ and $y_{i}=x_{i}$ for
$i\notin J$, in order to view $p$ as the new origin. Then%
\[
\begin{aligned} Z_{\Omega_{p}}^{\left( N\right) }(\boldsymbol{s})={\displaystyle\int\limits_{\mathbb{R}^{N-3}}} \Omega_{p}(y) & {\displaystyle\prod\limits_{i\notin J}} \left\vert y_{i}\right\vert ^{s_{1j}}\text{ } {\displaystyle\prod\limits_{i\in J}} \left\vert y_{i}\right\vert ^{s_{(N-1)j}}\text{ } \times \\ &{\prod\limits_{\substack{2\leq i<j\leq N-2\\i,j\in J}}}\left\vert y_{i}-y_{j}\right\vert ^{s_{ij}}\text{{ }}{\prod\limits_{\substack{2\leq i<j\leq N-2\\i,j\notin J}}}\left\vert y_{i}-y_{j}\right\vert ^{s_{ij}}\text{ }g_{p}\left( y,\boldsymbol{s}\right) {\displaystyle\prod\limits_{i=2}^{N-2}} dy_{i}, \label{Integral_1}\end{aligned}
\]
where $g_{p}\left(  y,\boldsymbol{s}\right)  $ is an invertible function on
the support of $\Omega_{p}(y)$, smooth in $y$ and holomorphic in
$\boldsymbol{s}$. This simplifies to
\begin{equation}
Z_{\Omega_{p}}^{\left(  N\right)  }(s)={\displaystyle\int\limits_{\mathbb{R}%
^{N-3}}}\Omega_{p}(y){\displaystyle\prod\limits_{i\in I}}\left\vert
y_{i}\right\vert ^{s}\text{ }{\prod\limits_{\substack{2\leq i<j\leq
N-2\\i,j\in J}}}\left\vert y_{i}-y_{j}\right\vert ^{s}\text{{ }}%
{\prod\limits_{\substack{2\leq i<j\leq N-2\\i,j\notin J}}}\left\vert
y_{i}-y_{j}\right\vert ^{s}\text{ }g_{p}\left(  y,s\right)
{\displaystyle\prod\limits_{i=2}^{N-2}}dy_{i}. \label{Integral_1bis}%
\end{equation}
The divisor $D_{p}$ attached to (\ref{Integral_1}) or (\ref{Integral_1bis}) is
given by the zero locus of
\[%
{\displaystyle\prod\limits_{i\in I}}
y_{i}\text{ }{\prod\limits_{\substack{2\leq i<j\leq N-2\\i,j\in J}}}\left(
y_{i}-y_{j}\right)  \text{ }{\prod\limits_{\substack{2\leq i<j\leq
N-2\\i,j\notin J}}}\left(  y_{i}-y_{j}\right)  .
\]
It can be considered as a subarrangement of the arrangement $D_{0}$. Hence, an
embedded resolution of $D_{p}$ can be constructed by (part of) the same
blow-ups we used to construct the embedded resolution of $D_{0}$. Take any
centre of blow-up $Z_{i}$, of codimension $v_{i}$, occurring in those
resolutions, leading to the exceptional submanifold $E_{i}$. Say $n_{i}$ and
$n_{i}^{\prime}$ are the number of hyperplanes in $D_{0}$ and $D_{p}$,
respectively, containing $Z_{i}$; then clearly $n_{i}^{\prime}\leq n_{i}$.
Hence the numerical data of $E_{i}$, considered in the embedded resolution of
$D_{0}$ and $D_{p}$, are $(n_{i},v_{i})$ and $(n_{i}^{\prime},v_{i})$,
respectively. Since $\frac{v_{i}}{n_{i}}\leq\frac{v_{i}}{n_{i}^{\prime}}$, all
new quotients of numerical data are again at least $\frac{2}{N-2}$. For the
original integral $Z_{\Omega_{p}}^{\left(  N\right)  }(\boldsymbol{s})$, we
get new conditions for convergence similar to (\ref{firstcondition}),
(\ref{secondcondition}), (\ref{thirdcondition}) till (\ref{lastcondition2}),
involving the same constant terms, but with other variables $s_{ij}$, still
all having coefficient $1$, and with \emph{at most} the number of variables
$s_{ij}$ as in those expressions.

In particular, if $\operatorname{Re}(s_{ij}) > -\frac{2}{N-2}$ for all $ij$,
then all convergence conditions are satisfied. (The \lq hardest\rq\ type is
(\ref{firstcondition}), when $\frac{(N-2)(N-3)}{2}$ variables $s_{ij}$ occur.)
\end{proof}

\subsection{\label{Case I different}Case $I\neq\left\{  2,\ldots,N-2\right\}
$}

In this case, recall that $Z_{I}^{(N)}(\boldsymbol{s})$ takes the form
\[
Z_{I}^{(N)}(\boldsymbol{s})={\displaystyle\int\limits_{\mathbb{R}^{N-3}}}
\text{ } \widetilde{\varphi}_{I}\left(  x\right)  \prod\limits_{i\not \in
I}\left\vert x_{i}\right\vert ^{-s_{1i}-s_{\left(  N-1\right)  i}%
-\sum_{\substack{2\leq j\leq N-2\\j\neq i}}s_{ij}-2}\text{ } F_{I}\left(
x,\boldsymbol{s}\right)  {\displaystyle\prod\limits_{i=2}^{N-2}} dx_{i},
\label{Eq_4}%
\]
where
\begin{align*}
F_{I}\left(  x,\boldsymbol{s}\right)   &  :=\prod\limits_{j\in I}\left\vert
x_{j}\right\vert ^{s_{1j}}\text{ }\prod\limits_{j=2}^{N-2}\left\vert
1-x_{j}\right\vert ^{s_{\left(  N-1\right)  j}}\prod\limits_{\substack{2\leq
i<j\leq N-2\\i,\text{ }j\in I}}\left\vert x_{i}-x_{j}\right\vert ^{s_{ij}%
}\text{ }\times\\
&  \prod\limits_{\substack{2\leq i<j\leq N-2\\i,\text{ }j\not \in
I}}\left\vert x_{i}-x_{j}\right\vert ^{s_{ij}}\text{ }\prod
\limits_{\substack{2\leq i<j\leq N-2\\i\not \in I,\text{ }j\in I}}\left\vert
1-x_{i}x_{j}\right\vert ^{s_{ij}}\prod\limits_{\substack{2\leq i<j\leq
N-2\\i\in I,\text{ }j\not \in I}}\left\vert 1-x_{i}x_{j}\right\vert ^{s_{ij}}.
\end{align*}
This simplifies to
\[
Z_{I}^{(N)}(s):=\int\limits_{\mathbb{R}^{N-3}}\widetilde{\varphi}_{I}\left(
x\right)  \prod\limits_{i\not \in I}\left\vert x_{i}\right\vert ^{-(N-2)s-2}%
\text{ } F_{I}\left(  x,s\right)
{\displaystyle\prod\limits_{i=2}^{N-2}}
dx_{i}.
\]
The relevant arrangement here is the divisor $D_{I,0} \cup D_{I}$, where
$D_{I,0} :=\{\prod_{i\notin I}x_{i}=0\}$ and $D_{I}$ is the zero locus of
\begin{align*}
\prod\limits_{i\in I}x_{i}\text{ }\prod\limits_{i=2}^{N-2}\left(
1-x_{i}\right)   &  \prod\limits_{\substack{2\leq i<j\leq N-2\\i,\text{
}j\notin I}}\left(  x_{i}-x_{j}\right)  \text{ }\prod\limits_{\substack{2\leq
i<j\leq N-2\\i,\text{ }j\in I}}\left(  x_{i}-x_{j}\right)  \text{ }\times\\
&  \prod\limits_{\substack{2\leq i<j\leq N-2\\i\not \in I,\text{ }j\in
I}}\left(  1-x_{i}x_{j}\right)  \prod\limits_{\substack{2\leq i<j\leq
N-2\\i\in I,\text{ }j\not \in I}}\left(  1-x_{i}x_{j}\right)  .
\end{align*}

\begin{proposition}
\label{Prop2}Fix $I\neq\left\{  2,\ldots,N-2\right\}  $. Then $Z_{I}%
^{(N)}(\boldsymbol{s})$ is convergent and holomorphic in a region containing
$-\frac{2}{N-2} < \operatorname{Re}(s_{ij}) < -\frac{2}{N} $ for all $ij$. In
particular, the integral $Z_{I}^{(N)}(s)$ is holomorphic in the band
$-\frac{2}{N-2} <\operatorname{Re}(s)<-\frac{2}{N}$.
\end{proposition}

\begin{proof}
Note first that in this case the divisors $D_{I}$ and $D_{I,0}$ play a
different role. All components of (the strict transforms of) $D_{I}$ induce
the classical condition $\operatorname{Re}(s_{ij})+1>0\Leftrightarrow
\operatorname{Re}(s_{ij})>-1$. But the components of $D_{I,0}$ give rise to a
new kind of convergence condition of type%
\begin{equation}
-\operatorname{Re}(s_{1i})-\operatorname{Re}(s_{\left(  N-1\right)  i}%
)-\sum_{\substack{2\leq j\leq N-2\\j\neq i}}\operatorname{Re}(s_{ij})-1>0,
\label{strict transform bound}%
\end{equation}
simplifying for $Z_{I}^{(N)}(s)$ to
\begin{equation}
-(N-2)\operatorname{Re}(s)-1>0\Leftrightarrow\operatorname{Re}(s)<-\frac
{1}{N-2}. \label{strict transform bound simple}%
\end{equation}
Next, we construct an embedded resolution of $D_{I}\cup D_{I,0}$. A crucial
observation is that any blow-up with centre \emph{not} contained in $D_{I,0}$
will induce a convergence condition that already appeared in the construction
of the resolution $\sigma$ in the proof of Proposition \ref{Prop1}. In
particular, when $-\frac{2}{N-2}<\operatorname{Re}(s_{ij})$ for all $ij$, then
all such conditions are satisfied. For $Z_{I}^{(N)}(s)$, these conditions get
the form $-\frac{2}{N-2}<\operatorname{Re}(s)$ or weaker. We could make this
lower bound more precise, depending on the size of $I$, but this would not
affect the end result of Theorem \ref{TheoremB}.

\smallskip We now treat the blow-ups with centre in $D_{I,0}$. In particular,
we will show that they induce for $Z_{I}^{(N)}(s)$ as strongest condition
$\operatorname{Re}(s)<-\frac{2}{N}$.

In a small enough neighborhood of $D_{I,0}$, we can write the integrand in
$Z_{I}^{(N)}(\boldsymbol{s})$ in the form
\begin{align*}
\widetilde{\varphi}_{I}\left(  x\right)   &  \prod\limits_{i\not \in
I}\left\vert x_{i}\right\vert ^{-s_{1i}-s_{\left(  N-1\right)  i}%
-\sum_{\substack{2\leq j\leq N-2\\j\neq i}}s_{ij}-2}\text{ } \prod
\limits_{\substack{2\leq i<j\leq N-2\\i,\text{ }j\not \in I}}\left\vert
x_{i}-x_{j}\right\vert ^{s_{ij}}\\
&  \prod\limits_{j\in I}\left\vert x_{j}\right\vert ^{s_{1j}}\text{ }%
\prod\limits_{j=2}^{N-2}\left\vert 1-x_{j}\right\vert ^{s_{\left(  N-1\right)
j}} \prod\limits_{\substack{2\leq i<j\leq N-2\\i,\text{ }j\in I}}\left\vert
x_{i}-x_{j}\right\vert ^{s_{ij}}\text{ } g(x,\boldsymbol{s}),
\end{align*}
where the factor $g(x,\boldsymbol{s})$ is invertible on the support of
$\widetilde{\varphi}_{I}\left(  x\right)  $ (and can be neglected from the
point of view of convergence and holomorphy). This simplifies for the
integrand in $Z_{I}^{(N)}(s)$ to
\begin{align*}
\widetilde{\varphi}_{I}\left(  x\right)   &  \text{{}}{\prod\limits_{i\notin
I}}\left\vert x_{i}\right\vert ^{-(N-2)s-2}\prod\limits_{\substack{2\leq
i<j\leq N-2\\i,\text{ }j\not \in I}}\left\vert x_{i}-x_{j}\right\vert
^{s}\text{ }\times\\
&  \prod\limits_{i\in I}\left\vert x_{i}\right\vert ^{s}\text{ }%
\prod\limits_{i\in I}\left\vert 1-x_{i}\right\vert ^{s}\prod
\limits_{\substack{2\leq i<j\leq N-2\\i,\text{ }j\in I}}\left\vert x_{i}%
-x_{j}\right\vert ^{s}\text{ }g(x,s).
\end{align*}

After a permutation of the indices, we may assume that $I=\left\{
2,\ldots,N-2\right\}  \setminus\left\{  2,\ldots,l\right\}  $ with $l\geq2$.
Then $D_{I,0}$ is given by $\prod_{2\leq i\leq l}x_{i}=0$. When $l=2$ (i.e.,
$|I| = N-4$), no blow-up with centre in $D_{I,0}$ is needed. If $l\geq3$
(i.e., $|I| \leq N-5$), we start by performing a blow-up $\tau$ with centre at
$x_{2}=\ldots=x_{l}=0$, for instance in the chart $x_{2}=u_{2}$, $x_{i}%
=u_{i}u_{2}$ for $i=3,\ldots,l$, and $x_{i}=u_{i}$ for $i=l+1,\ldots,N-2$.
This centre is contained in the $l-1$ hyperplanes $x_{i}=0$, $2\leq i\leq l$,
and in the corresponding $\frac{(l-1)(l-2)}{2}$ hyperplanes $x_{i}-x_{j}=0$,
with $2\leq i,j\leq l$. Hence the power of $|u_{2}|$ in the pullback of the
integrand in $Z_{I}^{(N)}(\boldsymbol{s})$ is
\[
\sum_{i=2}^{l} ( -s_{1i}-s_{\left(  N-1\right)  i}-\sum_{\substack{2\leq j\leq
N-2\\j\neq i}}s_{ij}-2 ) + \sum_{2\leq i<j \leq l} s_{ij} + (l-2) ,
\]
which we rewrite as
\begin{equation}
-\sum_{i=2}^{l} ( s_{1i}+s_{\left(  N-1\right)  i} ) -\sum_{i=2}^{l}
\sum_{j=l+1}^{N-2} s_{ij} - \sum_{2\leq i<j \leq l} s_{ij} - l .
\end{equation}
Note the important fact that all coefficients of the variables are $-1$. The
clue in the rewriting above is that, for $s_{ij}$ with $2\leq i<j \leq l$, we
had twice the coefficient $+1$ (coming from $x_{i}$ and $x_{j}$), and once the
coefficient $-1$ (coming from $x_{i}-x_{j}$). The associated new condition for
convergence for $Z_{I}^{(N)}(\boldsymbol{s})$ is thus
\begin{equation}
-\sum_{i=2}^{l} \big( \operatorname{Re}(s_{1i})+\operatorname{Re}(s_{\left(
N-1\right)  i} )\big) -\sum_{i=2}^{l} \sum_{j=l+1}^{N-2} \operatorname{Re}%
(s_{ij}) - \sum_{2\leq i<j \leq l} \operatorname{Re}(s_{ij}) - l+1 > 0 .
\label{EQ 19}%
\end{equation}
Note that $2(l-1) + (l-1)(N-l-2) + \frac{(l-1)(l-2)}{2} = (l-1)\frac
{(2N-l-2)}{2}$ different variables occur. Hence, for $Z_{I}^{(N)}(s)$, this
simplifies to $l-1)\frac{(2N-l-2)}{2} \operatorname{Re}(s)-(l-1)>0,$ i.e.,
\begin{equation}
\operatorname{Re}(s)<-\frac{2}{2N-l-2}. \label{EQ_19bis}%
\end{equation}
This condition is the strongest when $l=N-2$ (i.e., $I=\emptyset$), and
becomes then $\operatorname{Re}(s)<-\frac{2}{N}$. In that case (\ref{EQ 19})
is
\begin{equation}
-\sum_{i=2}^{N-2} \big( \operatorname{Re}(s_{1i})+\operatorname{Re}(s_{\left(
N-1\right)  i} )\big) - \sum_{2\leq i<j \leq N-2} \operatorname{Re}(s_{ij}) -
N+3 > 0 . \label{EQ l=N-2}%
\end{equation}

Further blow-ups to make the integrand locally monomial, for instance with
centre $u_{3}=\dots=u_{l}=0$ in the chart above, yield completely similar
conditions. In fact, after another permutation of the indices, we obtain again
the condition (\ref{EQ 19}), but for smaller $l$.

In particular, if $\operatorname{Re}(s_{ij}) < -\frac{2}{N} $ for all $ij$,
then all convergence conditions (\ref{strict transform bound}) and type
(\ref{EQ 19}) are satisfied. (The \lq hardest\rq\ type is (\ref{EQ l=N-2}),
due to the largest number of occurring variables.) For $Z_{I}^{(N)}(s)$, it is
indeed clear that $\operatorname{Re}(s)<-\frac{2}{N}$ is a stronger condition
than (\ref{strict transform bound simple}).
\end{proof}

\subsection{Proof of Theorem \ref{TheoremB} and precise description of the
convergence domain}

The statement of Theorem \ref{TheoremB} now follows quite immediately.

\begin{proof}
From Propositions \ref{Prop1} and \ref{Prop2}, we know that the Koba-Nielsen
local zeta function $Z_{\mathbb{R}}^{(N)}(\boldsymbol{s})$ is convergent and
holomorphic in a region containing $-\frac{2}{N-2}<\operatorname{Re}%
(s_{ij})<-\frac{2}{N}$ for all $ij$. This is enough to imply meromorphic
continuation to the whole $\mathbb{C}^{N-3}$.

The restrictions concerning the numbers $N_{ij,k}\left(  I\right)  $ and
$\gamma_{k}\left(  I\right)  $ in the actual convergence domain $\cap
_{I}\mathcal{H}(I)$, see (\ref{Eq_7}), are a consequence of the proofs of
these propositions.
\end{proof}

Moreover, the details in the strategy and outline of the proofs of
Propositions \ref{Prop1} and \ref{Prop2} yield the following more precise
description of the actual convergence domain. Below, the inequalities
(\ref{EQ B}) generalize (\ref{firstcondition}), (\ref{secondcondition}),
\dots, (\ref{lastcondition1}), where the set $J$ is $\{2,\dots,N-2\}$,
$\{3,\dots,N-2\}$,\dots, $\{N-3,N-2\}$, respectively. The inequalities
(\ref{EQ C}) are the similar ones according to the symmetry $1\leftrightarrow
N-1$ in the integrand of $Z_{\mathbb{R}}^{(N)}(\boldsymbol{s})$, corresponding
geometrically to $(0,0,\dots,0)\leftrightarrow(1,1,\dots,1)$. Next, the
inequalities (\ref{EQ D}) generalize (\ref{thirdcondition}),\dots,
(\ref{lastcondition2}), where $J$ is $\{2,\dots,N-2\}$, $\{3,\dots,N-2\}$,
\dots, $\{N-4,N-3,N-2\}$, respectively. Finally, in (\ref{EQ E}) the sets $J$
generalize $\{i\}$ in (\ref{strict transform bound}) and $\{2,\dots,l\}$ in
(\ref{EQ 19}).

\begin{proposition}
\label{exact convergence domain} The Koba-Nielsen local zeta function
$Z_{\mathbb{R}}^{(N)}(\boldsymbol{s})$ is convergent and holomorphic in the
region determined by the following $2^{N-1}-N-1$ inequalities:
\begin{equation}
\operatorname{Re}(s_{ij})>-1\label{EQ A}%
\end{equation}
for all $\frac{N(N-3)}{2}$ variables $s_{ij}$,
\begin{equation}
\sum_{j\in J}\operatorname{Re}(s_{1j})+\sum_{\substack{i,j\in J\\i<j}%
}\operatorname{Re}(s_{ij})>-\sharp J\label{EQ B}%
\end{equation}
for all subsets $J\subset\{2,\dots,N-2\}$ with $\sharp J\geq2$,
\begin{equation}
\sum_{j\in J}\operatorname{Re}(s_{(N-1)j})+\sum_{\substack{i,j\in
J\\i<j}}\operatorname{Re}(s_{ij})>-\sharp J\label{EQ C}%
\end{equation}
for all subsets $J\subset\{2,\dots,N-2\}$ with $\sharp J\geq2$,
\begin{equation}
\sum_{\substack{i,j\in J\\i<j}}\operatorname{Re}(s_{ij})>-\sharp
J+1\label{EQ D}%
\end{equation}
for all subsets $J\subset\{2,\dots,N-2\}$ with $\sharp J\geq3$,
\begin{equation}
\sum_{j\in J}\operatorname{Re}(s_{1j})+\sum_{j\in J}\operatorname{Re}%
(s_{\left(  N-1\right)  j})+\sum_{\substack{i\in\{2,\dots,N-2\}\setminus
J\\j\in J}}\operatorname{Re}(s_{ij})+\sum_{\substack{i,j\in J\\i<j}%
}\operatorname{Re}(s_{ij})<-\sharp J\label{EQ E}%
\end{equation}
for all subsets $J\subset\{2,\dots,N-2\}$ with $\sharp J\geq1$.
\end{proposition}


\subsection{Example\label{Example_N_6}}

Fix $N=6$. Following the proofs of Propositions \ref{Prop1} and \ref{Prop2},
or as special case of Proposition \ref{exact convergence domain}, we obtain
that $Z^{(6)}(s)$ is holomorphic in the region
\[
\begin{aligned}
\text{(basic)}\quad   & \operatorname{Re}(s_{ij}) > -1 \quad\text{for all nine } ij ,\\
(I=\{2,3.4\})\quad  & \operatorname{Re}(s_{12})+ \operatorname{Re}(s_{13})  + \operatorname{Re}(s_{14}) +\operatorname{Re}(s_{23}) + \operatorname{Re}(s_{24}) + \operatorname{Re}(s_{34}) > -3, \\
& \operatorname{Re}(s_{52})+ \operatorname{Re}(s_{53})  + \operatorname{Re}(s_{54}) +\operatorname{Re}(s_{23}) + \operatorname{Re}(s_{24}) + \operatorname{Re}(s_{34}) > -3, \\
& \operatorname{Re}(s_{12})+ \operatorname{Re}(s_{13})+\operatorname{Re}(s_{23}) > -2, \quad \operatorname{Re}(s_{12})+\operatorname{Re}(s_{14})+\operatorname{Re}(s_{24}) > -2  , \\
& \operatorname{Re}(s_{13})+ \operatorname{Re}(s_{14})+\operatorname{Re}(s_{34}) > -2, \quad \operatorname{Re}(s_{52})+\operatorname{Re}(s_{53})+\operatorname{Re}(s_{23}) > -2  , \\
& \operatorname{Re}(s_{52})+ \operatorname{Re}(s_{54})+\operatorname{Re}(s_{24}) > -2, \quad \operatorname{Re}(s_{53})+\operatorname{Re}(s_{54})+\operatorname{Re}(s_{34}) > -2  , \\
& \operatorname{Re}(s_{23}) + \operatorname{Re}(s_{24}) + \operatorname{Re}(s_{34}) > -2, \\
(|I| <3)\quad   &  \operatorname{Re}(s_{12})+\operatorname{Re}(s_{52})+\operatorname{Re}(s_{23})+\operatorname{Re}(s_{24})< -1 , \\    &\operatorname{Re}(s_{13})+\operatorname{Re}(s_{53})+\operatorname{Re}(s_{23})+\operatorname{Re}(s_{34})  < -1 ,\\
&\operatorname{Re}(s_{14})+\operatorname{Re}(s_{54})+\operatorname{Re}(s_{24})+\operatorname{Re}(s_{34})  < -1 ,\\
&  \operatorname{Re}(s_{12})+\operatorname{Re}(s_{52})+\operatorname{Re}(s_{13})+\operatorname{Re}(s_{53})+\operatorname{Re}(s_{23})+\operatorname{Re}(s_{24})+
\operatorname{Re}(s_{34})< -2 , \\
&  \operatorname{Re}(s_{12})+\operatorname{Re}(s_{52})+\operatorname{Re}(s_{14})+\operatorname{Re}(s_{54})+\operatorname{Re}(s_{23})+\operatorname{Re}(s_{24})+
\operatorname{Re}(s_{34})< -2 , \\
&  \operatorname{Re}(s_{13})+\operatorname{Re}(s_{53})+\operatorname{Re}(s_{14})+\operatorname{Re}(s_{54})+\operatorname{Re}(s_{23})+\operatorname{Re}(s_{24})+
\operatorname{Re}(s_{34})< -2 , \\
& \sum_{ij}  \operatorname{Re}(s_{ij}) < -3 ,
\end{aligned}
\]
containing the simpler domain $-\frac{1}{2}<\operatorname{Re}(s_{ij}%
)<-\frac{1}{3}$ for all $ij$.

\section{\label{Section_general_case}Local Zeta Functions of Koba-Nielsen type
Over Local Fields}

The Koba-Nielsen local zeta functions introduced in Definition
\ref{Definition_K:N_Zeta_func} can be naturally defined over arbitrary local
fields of characteristic zero, i.e., $\mathbb{R}$, $\mathbb{C}$, or finite
extensions of $\mathbb{Q}_{p}$, the field of $p$-adic numbers, and the proof
of the main theorem can be extended easily to the case of local fields
different from $\mathbb{R}$. We denote the corresponding local zeta functions
as $Z_{\mathbb{K}}^{(N)}(\boldsymbol{s})$ to emphasize the dependency on
$\mathbb{K}$. Note that the $p$-adic case was already treated in
\cite{Zun-B-C-LMP}-\cite{Zun-B-C-JHEP} through an alternative method, only
available in that case, called Igusa's stationary phase formula.

\subsection{Local fields}

We take $\mathbb{K}$ to be a non-discrete$\ $locally compact field of
characteristic zero. Then $\mathbb{K}$ is $\mathbb{R}$, $\mathbb{C}$, or a
finite extension of $\mathbb{Q}_{p}$, the field of $p$-adic numbers. If
$\mathbb{K}$ is $\mathbb{R}$ or $\mathbb{C}$, we say that $\mathbb{K}$ is an
$\mathbb{R}$\textit{-field}, otherwise we say that $\mathbb{K}$ is a
$p$\textit{-field}.

For $a\in\mathbb{K}$, we define the \textit{modulus} $\left\vert a\right\vert
_{\mathbb{K}}$ of $a$ by%
\[
\left\vert a\right\vert _{\mathbb{K}}=\left\{
\begin{array}
[c]{l}%
\text{the rate of change of the Haar measure in }(\mathbb{K},+)\text{ under
}x\rightarrow ax\text{ }\\
\text{for }a\neq0,\\
\\
0\ \text{ for }a=0\text{.}%
\end{array}
\right.
\]
It is well-known that, if $\mathbb{K}$ is an $\mathbb{R}$-field, then
$\left\vert a\right\vert _{\mathbb{R}}=\left\vert a\right\vert $ and
$\left\vert a\right\vert _{\mathbb{C}}=\left\vert a\right\vert ^{2}$, where
$\left\vert \cdot\right\vert $ denotes the usual absolute value in
$\mathbb{R}$ or $\mathbb{C}$, and, if $\mathbb{K}$ is a $p$-field, then
$\left\vert \cdot\right\vert _{\mathbb{K}}$ is the normalized absolute value
in $\mathbb{K}$.

We now take $\mathbb{K}$ to be a $p$-field. Let $R_{\mathbb{K}}$\ be the
valuation ring of $\mathbb{K}$, $P_{\mathbb{K}}$ the maximal ideal of
$R_{\mathbb{K}}$, and $\overline{\mathbb{K}}=R_{\mathbb{K}}/P_{\mathbb{K}}$
\ the residue field of $\mathbb{K}$. The cardinality of the residue field of
$\mathbb{K}$ is denoted by $q$, thus $\overline{\mathbb{K}}=\mathbb{F}_{q}$.
For $z\in\mathbb{K}$, $ord\left(  z\right)  \in\mathbb{Z}\cup\{+\infty\}$
\ denotes the valuation of $z$, and $\left\vert z\right\vert _{\mathbb{K}%
}=q^{-ord\left(  z\right)  }$. We fix a uniformizing parameter $\mathfrak{p}$
of $R_{\mathbb{K}}$, i.e., a generator of $P_{\mathbb{K}}$.

We fix a set $S_{\mathbb{K}}\subset R_{\mathbb{K}}$ of representatives \ of
the residue field $\overline{\mathbb{K}}$. We assume that $0\in S_{\mathbb{K}%
}$. Any $z\in\mathbb{K}\backslash\left\{  0\right\}  $ admits a power
expansion of the form%
\begin{equation}
z=\mathfrak{p}^{m}\sum\limits_{k=0}^{\infty}z_{k}\mathfrak{p}^{k}\text{,}
\label{Power_Exp}%
\end{equation}
where $m\in\mathbb{Z}$, the $z_{k}$ belong to $S_{\mathbb{K}}$, and $z_{0}%
\neq0$. The series (\ref{Power_Exp}) converges in the norm $\left\vert
\cdot\right\vert _{\mathbb{K}}$.

\subsection{Multivariate Local Zeta Functions: General Case}

If $\mathbb{K}$ is a $p$-field, resp. an $\mathbb{R}$-field, we denote by
$\mathcal{D}(\mathbb{K}^{n})$ the $\mathbb{C}$-vector space consisting of all
$\mathbb{C}$-valued locally constant functions, resp. all smooth functions, on
$\mathbb{K}^{n}$, with compact support. An element of $\mathcal{D}%
(\mathbb{K}^{n})$ is called a \textit{test function}.

Let $f_{1}(x),\ldots,f_{m}(x)\in\mathbb{K}\left[  x_{1},\ldots,x_{n}\right]  $
be non-constant polynomials, we denote by $D_{\mathbb{K}}:=\cup_{i=1}^{m}%
f_{i}^{-1}(0)$ the divisor attached to them. We set
\[
\boldsymbol{f}:=\left(  f_{1},\ldots,f_{m}\right)  \text{ \ and \ }%
\boldsymbol{s}:=\left(  s_{1},\ldots,s_{m}\right)  \in\mathbb{C}^{m}\text{.}%
\]
The multivariate local zeta function attached to $(\boldsymbol{f},\Theta)$,
with $\Theta\in\mathcal{D}(\mathbb{K}^{n})$, is defined as%
\begin{equation}
Z_{\Theta}\left(  \boldsymbol{f},\boldsymbol{s}\right)  =\int
\limits_{\mathbb{K}^{n}\smallsetminus D_{\mathbb{K}}}\Theta\left(  x\right)
\prod\limits_{i=1}^{m}\left\vert f_{i}(x)\right\vert _{\mathbb{K}}^{s_{i}}%
{\displaystyle\prod\limits_{i=1}^{n}}
dx_{i}\text{, \qquad when }\operatorname{Re}(s_{i})>0\text{ for all }i\text{.}
\label{Zeta Function general}%
\end{equation}
Integrals of type (\ref{Zeta Function general}) are analytic functions, and
they admit meromorphic continuations to the whole $\mathbb{C}^{m}$, see
\cite{Igusa-old},\cite{Igusa}, \cite{Kashiwara-Takai}, \cite{Loeser}. By
applying Hironaka's resolution of singularities theorem to $D_{\mathbb{K}}$,
the study of integrals of type (\ref{Zeta Function general}) is reduced to the
case of monomial integrals, which can be studied directly, see e.g.
\cite{Loeser}, \cite{Igusa}, \cite{Igusa-old}.

\begin{lemma}
\label{Lemma10} Let $\Phi\left(  y,s_{1},\ldots,s_{m}\right)  $ be a test
function with support in the polydisc
\[
\left\{  y\in\mathbb{K}^{n};\left\vert y_{i}\right\vert <1\text{, for
}i=1,\ldots,n\right\}  ,
\]
when $\mathbb{K}$ is an $\mathbb{R}$-field, and with support $\mathfrak{p}%
^{e}R_{\mathbb{K}}^{n}$ ($e\in\mathbb{Z}$) when $\mathbb{K}$ is a $p$-field,
which is holomorphic in $s_{1},\ldots,s_{m}$. Consider the integral
\[
J_{\mathbb{K}}(s_{1},\ldots,s_{m})=%
{\textstyle\int\limits_{\mathbb{K}^{n}}}
\Phi\left(  y,s_{1},\ldots,s_{m}\right)  \prod_{i=1}^{r}\left\vert
y_{i}\right\vert _{\mathbb{K}}^{\sum_{j=1}^{m}a_{{j},i}s_{j}+b_{i}-1}%
{\displaystyle\prod\limits_{i=1}^{n}}
dy_{i}\text{, }%
\]
where $1\leq r\leq n$, for each $i$ the $a_{j,i}$ are integers (not all zero)
and $b_{i}$ is an integer.\ Set%
\[
\mathcal{R}_{\mathbb{K}}:=%
{\textstyle\bigcap\limits_{i\in\left\{  1,\ldots,r\right\}  }}
\left\{  (s_{1},\ldots,s_{m})\in\mathbb{C}^{m};\sum_{j=1}^{m}a_{j,i}%
\operatorname{Re}(s_{j})+b_{i}>0\right\}  .
\]
Then the following assertions hold:

\noindent(i) if all the $a_{j,i}$ are nonnegative integers (not all zero) and
$b_{i}$ is a positive integer, then $\mathcal{R}_{\mathbb{K}}\neq\emptyset$.
More precisely, $\left\{  (s_{1},\ldots,s_{m})\in\mathbb{C}^{m}%
;\operatorname{Re}(s_{j})>0\text{, }j=1,\cdots,m\right\}  \subset
\mathcal{R}_{\mathbb{K}}$;

\noindent(ii) if $\mathcal{R}_{\mathbb{K}}\neq\emptyset$, then $J_{\mathbb{K}%
}(s_{1},\ldots,s_{m})$ is convergent and defines a holomorphic function in the
domain $\mathcal{R}_{\mathbb{K}}$. Moreover, in the $p$-field case, it is a
rational function in $q^{-s_{1}},\ldots,q^{-s_{m}}$;

\noindent(ii) if $\mathcal{R}_{\mathbb{K}}\neq\emptyset$, then the function
$J_{\mathbb{K}}(s_{1},\ldots,s_{m})$ admits an analytic continuation to the
whole $\mathbb{C}^{m}$, as a meromorphic function with poles belonging to
\[
\bigcup_{1\leq i\leq r}\text{ }\bigcup_{t}\left\{  \sum_{j=1}^{m}a_{j,i}%
s_{j}+b_{i}+t=0\right\}  ,
\]
in the $\mathbb{R}$-field case, with $t\in\mathbb{N}$ if $\mathbb{K}%
=\mathbb{R}$ and $t\in\frac{1}{2}\mathbb{N}$ if $\mathbb{K}=\mathbb{C}$, and
with poles belonging to
\[
\bigcup_{1\leq i\leq r}\text{ }\left\{  \sum_{j=1}^{m}a_{j,i}\operatorname{Re}%
(s_{j})+b_{i}=0\right\}  ,
\]
in the $p$-field case.
\end{lemma}

Also, Remark \ref{rem: holom} extends to this more general setting.

\begin{remark}
\label{Nota_Lemma_10}Theorem \ref{thm: num data and poles} then extends in an
obvious way to this more general setting. In\textrm{ }addition, in the
$p$-field case, the integral $Z_{\Theta}\left(  \boldsymbol{f},\boldsymbol{s}%
\right)  $ admits a meromorphic continuation as a rational function
\[
Z_{\Theta}\left(  \boldsymbol{f},\boldsymbol{s}\right)  =\frac{P_{\Theta
}\left(  \boldsymbol{s}\right)  }{\prod\limits_{i\in T}\left(  1-q^{-\left(
{\textstyle\sum\limits_{j=1}^{m}}
N_{f_{j},i}s_{j}+v_{i}\right)  }\right)  }%
\]
in $q^{-s_{1}},\ldots,q^{-s_{m}}$, where $P_{\Theta}\left(  \boldsymbol{s}%
\right)  $ is a polynomial in the variables $q^{-s_{i}}$, and the real parts
of its poles belong to the finite union of hyperplanes
\begin{equation}%
{\textstyle\sum\limits_{j=1}^{m}}
N_{f_{j},i}s_{j}+v_{i}=0,\quad\text{ for }i\in T, \label{INE_1}%
\end{equation}
cf. \cite[Th\'{e}or\`{e}me 1.1.4.]{Loeser}.
\end{remark}

\subsection{Meromorphic Continuation of Local Zeta Functions: General Case}

\begin{theorem}
\label{TheoremA} Let $\mathbb{K}$ be a local field of characteristic zero. The
Koba-Nielsen local zeta function $Z_{\mathbb{K}}^{(N)}(\boldsymbol{s})$ is a
holomorphic function in the solution set $\cap_{I}\mathcal{H}(I)$, see
(\ref{Eq_7}), in $\mathbb{C}^{\mathbf{d}}$, which contains the set
\[
-\frac{2}{N-2}<\operatorname{Re}(s_{ij})<-\frac{2}{N}\qquad\text{ for all
}ij\text{.}%
\]
Furthermore, it has a meromorphic continuation, denoted again as
$Z_{\mathbb{K}}^{(N)}(\boldsymbol{s})$, to the whole $\mathbb{C}^{\mathbf{d}}%
$. If $\mathbb{K}$ is an $\mathbb{R}$-field, the poles belong to $\cup
_{I}\mathcal{P}(I)$, see (\ref{Eq_7A}), where now $t\in\frac{1}{2}\mathbb{N}$
for $\mathbb{K}=\mathbb{C}$. If $\mathbb{K}$ is a $p$-field, then this
meromorphic continuation is a rational function in the variables $q^{-s_{ij}}%
$, with poles having real parts belonging to
\[
\bigcup\limits_{k\in T(I)}\left\{  s_{ij}\in\mathbb{R}^{\mathbf{d}}%
;\sum_{ij\in M\left(  I\right)  }N_{ij,k}(I)s_{ij}+\gamma_{k}(I)=0\right\}  ,
\]
where $N_{ij,k}\left(  I\right)  ,\gamma_{k}\left(  I\right)  \in\mathbb{Z}$,
and $M(I)$, $T(I)$ are finite sets. More precisely, for each $k$, either all
numbers $N_{ij,k}\left(  I\right)  $ are equal to $0$ or $1$ and $\gamma
_{k}\left(  I\right)  >0$, or all numbers $N_{ij,k}\left(  I\right)  $ are
equal to $0$ or $-1$ and $\gamma_{k}\left(  I\right)  <0$.
\end{theorem}

The proof of Theorem \ref{TheoremA} is a slight variation of the proof of
Theorem \ref{TheoremB}. We just indicate the required modifications. The first
step is to express $Z_{\mathbb{K}}^{(N)}(\boldsymbol{s})$ as a finite sum of
multivariate local zeta functions, see (\ref{Eq_2A}). This requires
introducing an analogue of the functions $\chi$, see (\ref{Function_Chi}), and
$\varphi_{I}$, see (\ref{Function_phi}). We first define the analogue of
$\chi$ in the complex case. We recall that an element of $\mathcal{D}%
(\mathbb{C}^{n})$ is a $C^{\infty}$ function in the variables $z_{1}%
,\overline{z_{1}},\ldots,z_{n},\overline{z_{n}}$ (or in $\operatorname{Re}%
(z_{1}),\operatorname{Im}\left(  z_{1}\right)  ,\ldots,\operatorname{Re}%
(z_{n}),\operatorname{Im}\left(  z_{n}\right)  $). We pick, for $z=x+iy$
($x,y\in\mathbb{R}$),%
\[
\chi_{_{\mathbb{C}}}\left(  z\right)  :=\chi\left(  \left\vert z\right\vert
_{\mathbb{C}}\right)  =\chi\left(  x^{2}+y^{2}\right)  ,
\]
where $\chi$ is defined as in (\ref{Function_Chi}). Then $\chi_{_{\mathbb{C}}%
}$ is a $C^{\infty}$ function in the variables $x,y$ satisfying
\[
\chi_{_{\mathbb{C}}}\left(  z\right)  =\left\{
\begin{array}
[c]{lll}%
1 & \text{if} & 0\leq\left\vert z\right\vert _{\mathbb{C}}\leq2\\
&  & \\
0 & \text{if} & \left\vert z\right\vert _{\mathbb{C}}\geq2+\epsilon.
\end{array}
\right.
\]
We now define the function $\varphi_{I}$ as in (\ref{Function_phi}). In the
$p$-field case, we use
\[
\chi_{\mathfrak{p}}\left(  z\right)  =\left\{
\begin{array}
[c]{lll}%
1 & \text{if} & \left\vert z\right\vert _{\mathbb{K}}\leq1\\
&  & \\
0 & \text{if} & \left\vert z\right\vert _{\mathbb{K}}>1.
\end{array}
\right.
\]
Now the proof follows line by line the one given for Theorem \ref{TheoremB}.
This is possible because, for any $\mathbb{K}$, all the required blow-ups and
centres are defined over the field of rational numbers.

\begin{remark}
Note that the convergence domains of $Z_{\mathbb{R}}^{(N)}(s)$ and
$Z_{\mathbb{C}}^{(N)}(s)$ are exactly the same, due to the conventions
$\left\vert a\right\vert _{\mathbb{R}}=\left\vert a\right\vert $ ($a\in R$)
and $\left\vert a\right\vert _{\mathbb{C}}=\left\vert a\right\vert ^{2}$
($a\in C$) in the defining integrals. 

The actual convergence domain of $Z_{\mathbb{C}}^{(N)}(s)$ can thus be
described as in Proposition \ref{exact convergence domain}, and in particular
for $N=4,5,6$ as in Examples \ref{Example_N_4}, \ref{Example_N_5},
\ref{Example_N_6}, respectively.\textrm{ }
\end{remark}

\subsection{A result of Vanhove and Zerbini}

\label{Vanhove and Zerbini}

In \cite[Proposition 7.2]{Vanhove et al}, Vanhove and Zerbini also studied the
domain of convergence of
\[
Z_{\mathbb{C}}^{(N)}\left(  \boldsymbol{s}\right)  ={\int\limits_{\left(
\mathbb{P}_{\mathbb{C}}^{1}\right)  ^{N-3}}}%
{\displaystyle\prod\limits_{i=2}^{N-2}}
\left\vert x_{j}\right\vert ^{2s_{1j}}\left\vert 1-x_{j}\right\vert
^{2s_{(N-1)j}}\text{ }%
{\displaystyle\prod\limits_{2\leq i<j\leq N-2}}
\left\vert x_{i}-x_{j}\right\vert ^{2s_{ij}}%
{\displaystyle\prod\limits_{i=2}^{N-2}}
dx_{i}.
\]
Note that this integral is indeed our $Z_{\mathbb{C}}^{(N)}\left(
\boldsymbol{s}\right)  $,\ since $\left(  \mathbb{P}_{\mathbb{C}}^{1}\right)
^{N-3}$ differs from $\mathbb{C}^{N-3}$ only by a set of measure zero and
$\left\vert a\right\vert _{\mathbb{C}}=\left\vert a\right\vert ^{2}$ for
$a\in\mathbb{C}$. We claim that, for $N\geq5$, the convergence domain they
describe is too large. In order to explain this, we first compare their result
to ours in the illustrative case $N=5$, that is,%
\[
Z_{\mathbb{C}}^{(5)}\left(  \boldsymbol{s}\right)  ={\int\limits_{\left(
\mathbb{P}_{\mathbb{C}}^{1}\right)  ^{2}}}%
{\displaystyle\prod\limits_{j=2}^{3}}
\left\vert x_{j}\right\vert ^{2s_{1j}}\left\vert 1-x_{j}\right\vert ^{2s_{4j}%
}\text{ }\left\vert x_{2}-x_{3}\right\vert ^{2s_{23}}dx_{2}dx_{3}.
\]
We established in  Example \ref{Example_N_5}  that $Z_{\mathbb{C}}%
^{(5)}\left(  \boldsymbol{s}\right)  $ is \ convergent (and holomorphic) in
the domain%
\begin{equation}
\left\{
\begin{array}
[c]{l}%
{\operatorname{Re}}(s_{12})>-1,{\operatorname{Re}}(s_{13}%
)>-1,{\operatorname{Re}}(s_{42})>-1,{\operatorname{Re}}(s_{43}%
)>-1,{\operatorname{Re}}(s_{23})>-1;\\
{\operatorname{Re}}(s_{12})+\operatorname{Re}(s_{13})+\operatorname{Re}%
(s_{23})>-2;\\
\operatorname{Re}(s_{42})+\operatorname{Re}(s_{43})+\operatorname{Re}%
(s_{23})>-2;\\
\operatorname{Re}(s_{12})+\operatorname{Re}(s_{42})+\operatorname{Re}%
(s_{23})<-1;\\
\operatorname{Re}(s_{13})+\operatorname{Re}(s_{43})+\operatorname{Re}%
(s_{23})<-1;\\
\operatorname{Re}(s_{12}+s_{13}+s_{42}+s_{43}+s_{23})<-2,
\end{array}
\right.  \label{Domain_1}%
\end{equation}
Now, according to \cite{Vanhove et al}, the integral $Z_{\mathbb{C}}%
^{(5)}\left(  \boldsymbol{s}\right)  $ is convergent in the larger domain
\begin{equation}
\left\{
\begin{array}
[c]{l}%
\operatorname{Re}\left(  s_{12}\right)  >-1;\operatorname{Re}\left(
s_{13}\right)  >-1;\operatorname{Re}\left(  s_{42}\right)
>-1;\operatorname{Re}\left(  s_{43}\right)  >-1;\operatorname{Re}\left(
s_{23}\right)  >-1;\\
\operatorname{Re}\left(  s_{12}+s_{13}+s_{23}\right)  >-2\\
\operatorname{Re}\left(  s_{43}+s_{42}+s_{23}\right)  >-2\\
\operatorname{Re}\left(  s_{12}+s_{13}+s_{42}+s_{43}+s_{23}\right)  <-2,
\end{array}
\right.  \label{Domain_2}%
\end{equation}
i.e., (\ref{Domain_1}) without the two inequalities on the fourth and fifth
lines. In order to compare with \cite[Section 7.1]{Vanhove et al}, we mention
that they use the notation $\left(  a_{1},a_{2},b_{1},b_{2},c_{12}\right)  $
for $\left(  s_{12},s_{13},s_{42},s_{43},s_{23}\right)  $ and $k$ for $N-3$.
There is a small typo in the inequalities on the second line on their (7.8),
corresponding to (\ref{Domain_2}) above. According to the outline of the proof
of \cite[Proposition 7.2]{Vanhove et al}, the domain (\ref{Domain_2}) is
obtained by using the partition
\[
\mathbb{C}^{2}=\left\{  \left(  x_{2},x_{3}\right)  \in\mathbb{C}%
^{2};\left\vert x_{2}\right\vert \leq\left\vert x_{3}\right\vert \right\}
\bigsqcup\left\{  \left(  x_{2},x_{3}\right)  \in\mathbb{C}^{2};\left\vert
x_{3}\right\vert <\left\vert x_{2}\right\vert \right\}  ,
\]
and and performing a changes of variables of type $u=\frac{x_{2}}{x_{3}}$ and
$u^{\prime}=\frac{x_{3}}{x_{2}}$ in the first and second region, respectively.
In this way they derive the regions where the integral is convergent near
$(0,0)$ and at infinity; the region where the integral is convergent near
$(1,1)$ is then obtained by substituting $s_{1j}$ with $s_{4j}$.

However, the authors do not write any details about their computations, and
they forget some inequalities, probably because they do not analyse the
situation `at infinity'\ carefully. Their last inequality is needed for
convergence near $(\infty,\infty)$, and the two missing ones for convergence
near $(\infty,0)$ and $(0,\infty)$, respectively. For $(\infty,0)$, one must
perform the change of variable $x_{2}\rightarrow1/x_{2}$, leading in Example
\ref{Example_N_5} to the study of convergence of the integral
(\ref{first infinity}) around the origin, and hence to the first missing
condition (\ref{ex4}). Similarly, convergence near $(0,\infty)$ leads to the
second missing condition (\ref{ex5}).

Note that the domain (\ref{Domain_2}) is really larger than our
(\ref{Domain_1}); for example the point $\boldsymbol{s}_{0}:=(s_{12}%
,s_{13},s_{42},s_{43},s_{23})=(-1/4,-2/3,-2/3,-2/3,0)$ belongs to the
difference set. Then
\[
\begin{aligned}
Z_{\mathbb{C}}^{(5)}\left(  \boldsymbol{s}_0\right)&=
{\displaystyle\int\limits_{\mathbb{C}^{2}}}
\left\vert x_{2}\right\vert ^{2\cdot (-\frac14 )}
\left\vert x_{3}\right\vert ^{2\cdot (-\frac23) }
\left\vert 1-x_{2}\right\vert^{2 \cdot (-\frac 23)}
\left\vert 1-x_{3}\right\vert^{2 \cdot (-\frac 23)}
dx_{2}dx_{3}\\
&= {\displaystyle\int\limits_{\mathbb{C}}}
\left\vert x_{2}\right\vert ^{2\cdot (-\frac14 )}   \left\vert 1-x_{2}\right\vert^{2 \cdot (-\frac 23)} dx_{2} \cdot
{\displaystyle\int\limits_{\mathbb{C}}}
\left\vert x_{3}\right\vert ^{2\cdot (-\frac23) }   \left\vert 1-x_{3}\right\vert^{2 \cdot (-\frac 23)} dx_{3},
\end{aligned}
\]
should be a convergent integral according to \cite[Proposition 7.2 and
(7.8)]{Vanhove et al}. The integral in $x_{2}$ however does \emph{not}
converge. Indeed, a necessary condition for convergence of the integral
\[
{\int\limits_{\mathbb{C}}\left\vert x\right\vert ^{-2\alpha}\left\vert
1-x\right\vert ^{-2\beta}dx}%
\]
is that $\alpha+\beta>1$. This is in fact part of the description of the
convergence domain of $Z_{\mathbb{C}}^{(4)}\left(  \boldsymbol{s}\right)  $ in
Example \ref{Example_N_4} and in \cite[(7.7)]{Vanhove et al}, which indeed
does not contain $(-\frac{1}{4},-\frac{2}{3})$.

More generally, for $N\geq5$, the statement in \cite[Proposition 7.2]{Vanhove
et al} misses the $2^{N-3}-2$ inequalities in (\ref{EQ E}) with $J\subsetneq
\{1,\dots,N-3\}$. For $N=6$, these are the six inequalities of the form
$\dots<-1$ or $\dots<-2$ in Example \ref{Example_N_6}.

\section{Meromorphic Continuation of Koba-Nielsen String Amplitudes over local
fields of characteristic zero}

We set%
\begin{equation}
A_{\mathbb{K}}^{(N)}\left(  \boldsymbol{k}\right)  =Z_{\mathbb{K}}%
^{(N)}(\boldsymbol{s})\mid_{s_{ij}=\boldsymbol{k}_{i}\boldsymbol{k}_{j}},
\label{Eq_redidifinition_Amplitude}%
\end{equation}
where the momenta vectors $\boldsymbol{k}_{i}$ belong to $\mathbb{C}^{l+1}$,
where $l$ is an arbitrary positive integer. Typically $l$ is taken to be $25$.
By Theorem \ref{TheoremA}, $Z_{\mathbb{K}}^{(N)}(\boldsymbol{s})$ has a
meromorphic continuation to the whole $\mathbb{C}^{\boldsymbol{d}}$.
This section is dedicated to the study of $A_{\mathbb{K}}^{(N)}\left(
\boldsymbol{k}\right)  $ as a meromorphic function.

\begin{remark}
Strictly speaking $A_{\mathbb{K}}^{(N)}\left(  \boldsymbol{k}\right)  $,
$\boldsymbol{k=}\left(  \boldsymbol{k}_{1},\ldots,\boldsymbol{k}_{N}\right)
\in\mathbb{R}^{N(l+1)}$, is defined in the real affine set%
\begin{equation}
\left\{  \left(  \boldsymbol{k}_{1},\ldots,\boldsymbol{k}_{N}\right)
\in\mathbb{R}^{N(l+1)};\sum_{i=1}^{N}\boldsymbol{k}_{i}=\boldsymbol{0}\text{,
\ \ \ \ \ }\boldsymbol{k}_{i}\boldsymbol{k}_{i}=2\text{ \ for }i=1,\ldots
,N\right\}  , \label{Kineematic_restritions}%
\end{equation}
where the mass of the tachyons is normalized \ as $m^{2}=-2$, in the case of
open strings, and $m^{2}=-4$, in the case of closed strings. In Theorem
\ref{TheoremC} we discard the kinematic restrictions
(\ref{Kineematic_restritions}), and consider $A_{\mathbb{K}}^{(N)}\left(
\boldsymbol{k}\right)  $ as defined in $\boldsymbol{k=}\left(  \boldsymbol{k}%
_{1},\ldots,\boldsymbol{k}_{N-1}\right)  \in\mathbb{C}^{\left(  N-1\right)
(l+1)}$. At the end of this section we discuss the convergence of the string
amplitudes considering the kinematic restrictions
(\ref{Kineematic_restritions}).
\end{remark}

\subsection{Convergence of the Koba-Nielsen amplitudes}

Here we consider
\[
\boldsymbol{k}_{i}=(k_{0,i},k_{1,i},\ldots,k_{l,i})\in\mathbb{C}^{l+1}\text{,
for \ }i=1,2,\ldots,N-1.
\]
Given two real numbers $B$ and $C$ such that $0<B<C<\frac{1}{N-2}$ and
$C-B=\frac{1}{N}$, we define the open subset $\mathcal{U}\subset$
$\mathbb{C}^{\left(  N-1\right)  (l+1)}$ as
\begin{equation}
\mathcal{U}:=\mathcal{U}_{1}\times\cdots\times\mathcal{U}_{N-1},
\label{Definition_set_U}%
\end{equation}
where each $\mathcal{U}_{i}\subset$ $\mathbb{C}^{l+1}$ is itself an open
subset defined as%
\[
\mathcal{U}_{i}=\mathcal{U}_{0,i}\times\mathcal{U}_{1,i}\times\cdots
\times\mathcal{U}_{l,i},
\]
where
\[
\mathcal{U}_{0,i}:=\left\{  k_{0,i}\in\mathbb{C};\ \sqrt{C}<\operatorname{Re}%
(k_{0,i})<\sqrt{\frac{1}{N-2}}\quad\text{ and}\quad0<\mathrm{\operatorname{Im}%
}(k_{0,i})<\sqrt{B}\right\}
\]
and
\[
\mathcal{U}_{m,i}:=\left\{  k_{m,i}\in\mathbb{C};\ 0<\mathrm{\operatorname{Re}%
}(k_{m,i})<\sqrt{\frac{B}{l}}\quad\text{ and}\quad\sqrt{\frac{C}{l}%
}<\mathrm{\operatorname{Im}}(k_{m,i})<\sqrt{\frac{1}{l(N-2)}}\right\}  ,
\]
for $m=1,\dots,l$.

\begin{proposition}
\label{Lemma_Convergence}
All elements $(\boldsymbol{k}_{1},\dots, \boldsymbol{k}_{N-1}) $ in
$\mathcal{U}\subset$ $\mathbb{C}^{\left(  N-1\right)  (l+1)}$ satisfy
\begin{equation}
-\frac{2}{N-2}<\operatorname{Re}(\boldsymbol{k}_{i}\boldsymbol{k}_{j}%
)<-\frac{2}{N}\qquad\text{for all }i,j\in\{1,2,\ldots,N-1\}.
\label{Conclusion}%
\end{equation}

\end{proposition}

\begin{proof}
By definition of the Minkowski product, we have for all $i,j$ that
\[
\boldsymbol{k}_{i}\boldsymbol{k}_{j}=-k_{0,i}k_{0,j}+\sum_{m=1}^{l}%
k_{m,i}k_{m,j},
\]
and hence
\begin{align*}
\mathrm{\operatorname{Re}}(\boldsymbol{k}_{i}\boldsymbol{k}_{j})  &
=-\mathrm{\operatorname{Re}}(k_{0,i})\mathrm{\operatorname{Re}}(k_{0,j}%
)+\sum_{m=1}^{l}\mathrm{\operatorname{Re}}(k_{m,i})\mathrm{\operatorname{Re}%
}(k_{m,j})\\
&  +\mathrm{\operatorname{Im}}(k_{0,i})\mathrm{\operatorname{Im}}%
(k_{0,j})-\sum_{m=1}^{l}\mathrm{\operatorname{Im}}(k_{m,i}%
)\mathrm{\operatorname{Im}}(k_{m,j}).
\end{align*}
Now for any $(\boldsymbol{k}_{1},\ldots,\boldsymbol{k}_{N-1})\in\mathcal{U}$
we have
\[
0<\sum_{m=1}^{l}\mathrm{\operatorname{Re}}(k_{m,i})\mathrm{\operatorname{Re}%
}(k_{m,j})<B \quad\text{ and } \quad C<\mathrm{\operatorname{Re}}%
(k_{0,i})\mathrm{\operatorname{Re}}(k_{0,j})<\frac{1}{N-2}%
\]
and
\[
C<\sum_{m=1}^{l}\mathrm{\operatorname{Im}}(k_{m,i})\mathrm{\operatorname{Im}%
}(k_{m,j})<\frac{1}{N-2} \quad\text{ and } \quad0<\mathrm{\operatorname{Im}%
}(k_{0,i})\mathrm{\operatorname{Im}}(k_{0,j})<B,
\]
and thus
\begin{equation}
-\frac{1}{N-2}<-\mathrm{\operatorname{Re}}(k_{0,i})\mathrm{\operatorname{Re}%
}(k_{0,j})+\sum_{m=1}^{l}\mathrm{\operatorname{Re}}(k_{m,i}%
)\mathrm{\operatorname{Re}}(k_{m,j})<B-C=-\frac{1}{N} \label{D2}%
\end{equation}
and
\begin{equation}
-\frac{1}{N-2}<\mathrm{\operatorname{Im}}(k_{0,i})\mathrm{\operatorname{Im}%
}(k_{0,j})-\sum_{m=1}^{l}\mathrm{\operatorname{Im}}(k_{m,i}%
)\mathrm{\operatorname{Im}}(k_{m,j})<B-C=-\frac{1}{N}. \label{D3}%
\end{equation}
Finally (\ref{Conclusion}) follows from (\ref{D2}) and (\ref{D3}).
\end{proof}


\subsection{Meromorphic Continuation of Koba-Nielsen String Amplitudes}

\begin{theorem}
\label{TheoremC}Let $\mathbb{K}$ be a local field of characteristic zero. The
integral $A_{\mathbb{K}}^{(N)}\left(  \boldsymbol{k}\right)  $ converges and
is holomorphic in the open set $\mathcal{U}\subset\mathbb{C}^{\left(
N-1\right)  (l+1)}$. It extends to a meromorphic function in $\boldsymbol{k}$
on the whole $\mathbb{C}^{\left(  N-1\right)  (l+1)}$.

If $\mathbb{K}$ is an $\mathbb{R}$-field, then the possible poles of
$A_{\mathbb{K}}^{(N)}\left(  \boldsymbol{k}\right)  $ belong to
\begin{equation}
\label{polarlocusR}%
{\displaystyle\bigcup\limits_{I\subseteq\left\{  2,\ldots,N-2\right\}  }}
\bigcup\limits_{t\in\mathbb{N}}\bigcup\limits_{r\in T(I)}\left\{
\boldsymbol{k}\in\mathbb{C}^{\left(  N-1\right)  (l+1)};\sum_{ij\in M\left(
I\right)  }N_{ij,r}(I)\boldsymbol{k}_{i}\boldsymbol{k}_{j}+\gamma_{r}%
(I)+\frac{t}{\left[  \mathbb{K}:\mathbb{R}\right]  }=0\right\}  ,
\end{equation}
where $N_{ij,r}\left(  I\right)  ,\gamma_{r}\left(  I\right)  \in\mathbb{Z}$,
and $M(I)$, $T(I)$ are finite sets, and $\left[  \mathbb{K}:\mathbb{R}\right]
=1$ if $\mathbb{K=R}$, and $\left[  \mathbb{K}:\mathbb{R}\right]  =2$ if
$\mathbb{K=C}$. If $\mathbb{K}$ is a $p$-adic field, then $A_{\mathbb{K}%
}^{(N)}\left(  \boldsymbol{k}\right)  $ is a rational function in the
variables $q^{-\boldsymbol{k}_{i}\boldsymbol{k}_{j}}$, and its possible poles
belong to%
\[%
{\displaystyle\bigcup\limits_{I\subseteq\left\{  2,\ldots,N-2\right\}  }}
\bigcup\limits_{r\in T(I)}\left\{  \boldsymbol{k}\in\mathbb{C}^{\left(
N-1\right)  (l+1)};\sum_{ij\in M\left(  I\right)  }N_{ij,r}%
(I)\operatorname{Re}(\boldsymbol{k}_{i}\boldsymbol{k}_{j})+\gamma
_{r}(I)=0\right\}  .
\]
More precisely, for each $r$, either all numbers $N_{ij,r}\left(  I\right)  $
are equal to $0$ or $1$ and $\gamma_{r}\left(  I\right)  >0$, or all numbers
$N_{ij,r}\left(  I\right)  $ are equal to $0$ or $-1$ and $\gamma_{r}\left(
I\right)  <0$.
\end{theorem}

\begin{proof}
We consider the polynomial (and hence holomorphic) mapping
\[
M: \mathbb{C}^{\left(  N-1\right)  (l+1)} \to\mathbb{C}^{\boldsymbol{d}} =
\mathbb{C}^{\frac{(N(N-3)}{2}} : \boldsymbol{k}=\left(  \boldsymbol{k}%
_{1},\ldots,\boldsymbol{k}_{N-1}\right)  \mapsto\boldsymbol{s} = (s_{ij}),
\]
given by $s_{ij}= \boldsymbol{k}_{i}\boldsymbol{k}_{j}$.

By Theorem \ref{TheoremA} and Proposition \ref{Lemma_Convergence}, we have
that $M(\mathcal{U})$ is part of the region where the original integral
defining the Koba-Nielsen zeta function $Z_{\mathbb{K}}^{(N)}(\boldsymbol{s})$
converges. As a consequence, the original integral defining the Koba-Nielsen
amplitude $A_{\mathbb{K}}^{(N)}\left(  \boldsymbol{k}\right)  $ converges in
$\mathcal{U}$. Then, since a composition of holomorphic mappings is again
holomorphic, it is clear that $A_{\mathbb{K}}^{(N)}\left(  \boldsymbol{k}%
\right)  $ has a meromorphic continuation to the whole $\mathbb{C}^{\left(
N-1\right)  (l+1)}$, with polar locus contained in the inverse image by $M$ of
the polar locus of $Z_{\mathbb{K}}^{(N)}(\boldsymbol{s})$. The description of
this polar locus follows directly from Theorem \ref{TheoremA}.

The rationality of $A_{\mathbb{K}}^{(N)}\left(  \boldsymbol{k}\right)  $ in
the case of $p$-adic fields follows from Remark \ref{Nota_Lemma_10}.
\end{proof}

\subsection{\label{Section_Tachyon_Scattering}Tachyon scattering}

We now discuss the interaction of $N$ tachyons with momenta $\boldsymbol{k}%
_{1},\ldots,\boldsymbol{k}_{N}\in\mathbb{R}^{l+1}$. We assume that the mass of
each tachyon satisfies $m^{2}=-2$, for open strings, and $m^{2}=-4$, for
closed strings. In this framework, $\boldsymbol{k}_{i}\boldsymbol{k}_{i}=2$
becomes the relativistic formula for the energy of the $i$-th tachyon. Then
increasing $N$ means to increase the whole energy of the scattering process.
This study requires finding solutions of kinematic restrictions belonging to
the domain of convergence of the corresponding Koba-Nielsen integral, i.e.,
solutions of
\begin{equation}
\left\{
\begin{array}
[c]{lllllll}%
\boldsymbol{k}_{i}\boldsymbol{k}_{j} \in\left(  -\frac{2}{N-2},-\frac{2}%
{N}\right)  &  &  & \text{for} &  &  & \left\{
\begin{array}
[c]{l}%
i=1\text{, }j=2,\ldots,N-2\text{, or}\\
i=N-1\text{, }j=2,\ldots,N-2\text{, }\\
\text{or }2\leq i<j\leq N-2\text{;}%
\end{array}
\right. \\
&  &  &  &  &  & \\
\sum_{i=1}^{N}\boldsymbol{k}_{i}=\boldsymbol{0}; &  &  &  &  &  & \\
&  &  &  &  &  & \\
\boldsymbol{k}_{i}\boldsymbol{k}_{i}=2 &  &  & \text{for} &  &  &
i=1,\ldots,N.
\end{array}
\right.  \label{Eq_solution_set}%
\end{equation}
(Strictly speaking, the convergence domain is larger than the region in the
first line of (\ref{Eq_solution_set}). But this region is the realistic one to
investigate uniformly in $N$.) Our efforts for finding solutions of
(\ref{Eq_solution_set}) suggest that is unlikely to find solutions for $N$
large. In contrast, if $N\leq l+1$, then is a easy to find solutions of
(\ref{Eq_solution_set}). We interpret this as the fact that high energy
scattering processes are less probable that the ones with low energy.

\begin{proposition}
\label{Prop3}If $N\leq l+1$, then the conditions (\ref{Eq_solution_set}) have
solutions $\boldsymbol{k}_{1},\ldots,\boldsymbol{k}_{N}\in\mathbb{R}^{l+1}$.
\end{proposition}

\begin{proof}
We take $\boldsymbol{k}_{i}=\left(  t_{i},\overrightarrow{\boldsymbol{k}_{i}%
}\right)  $, where $t_{i}\in\mathbb{R}$ and $\overrightarrow{\boldsymbol{k}%
_{i}} \in\mathbb{R}^{l}$ for $i=1,\ldots,N$, so that $\boldsymbol{k}%
_{i}\boldsymbol{k}_{j}=-t_{i}t_{j}+\left\langle \overrightarrow{\boldsymbol{k}%
_{i}}\text{,}\overrightarrow{\boldsymbol{k}_{j}}\right\rangle $, where
$\left\langle \cdot,\cdot\right\rangle $ denotes the standard inner product on
$\mathbb{R}^{l}$. Then (\ref{Eq_solution_set}) becomes
\begin{equation}%
{\displaystyle\sum\limits_{i=1}^{N}}
t_{i}=0\text{, \ }%
{\displaystyle\sum\limits_{i=1}^{N}}
\overrightarrow{\boldsymbol{k}_{i}}=0\text{;} \label{System_1}%
\end{equation}%
\begin{equation}
2=-t_{i}^{2}+\left\langle \overrightarrow{\boldsymbol{k}_{i}}\text{,}%
\overrightarrow{\boldsymbol{k}_{i}}\right\rangle \text{ for }i=1,\ldots,N;
\label{System_2}%
\end{equation}%
\begin{equation}
-\frac{2}{N-2}<-t_{i}t_{j}+\left\langle \overrightarrow{\boldsymbol{k}_{i}%
}\text{,}\overrightarrow{\boldsymbol{k}_{j}}\right\rangle <-\frac{2}{N},
\label{System_3}%
\end{equation}
for $i$, $j$ as in the first line in (\ref{Eq_solution_set}).

We now take
\[
t_{i}=\sqrt{\frac{2}{N-1}}\text{, \ \ \ }\left\Vert \overrightarrow
{\boldsymbol{k}_{i}}\right\Vert =\sqrt{\left\langle \overrightarrow
{\boldsymbol{k}_{i}}\text{,}\overrightarrow{\boldsymbol{k}_{i}}\right\rangle
}=\sqrt{\frac{2N}{N-1}}\text{ \ for }i=1,\ldots,N-1,
\]
such that $\left\{  \overrightarrow{\boldsymbol{k}_{i}}\right\}  _{1\leq i\leq
N-1}$ are two by two orthogonal. Such a choice is possible since $N-1\leq l$,
for instance%
\[
\overrightarrow{\boldsymbol{k}_{i}}=\left(  0,\ldots,0,\pm\sqrt{\frac{2N}%
{N-1}},0,\ldots,0\right)  ,
\]
where $\pm\sqrt{\frac{2N}{N-1}}$ appears in the $i$-th coordinate. The
conditions (\ref{System_1}) dictate that $t_{N}=-\sum_{i=1}^{N-1}t_{i}$ and
$\overrightarrow{\boldsymbol{k}_{N}}=-\sum_{i=1}^{N-1}\overrightarrow
{\boldsymbol{k}_{i}}$. We claim that (\ref{System_2}) and (\ref{System_3}) are
satisfied. Condition (\ref{System_2}) is obvious for $i=1,\dots,N-1$. It is
also satisfied for $i=N$ since
\[
-(\sum_{i=1}^{N-1}t_{i})^{2}+\langle\sum_{i=1}^{N-1}\overrightarrow
{\boldsymbol{k}_{i}},\sum_{i=1}^{N-1}\overrightarrow{\boldsymbol{k}_{i}%
}\rangle=-(N-1)2+(N-1)\frac{2N}{N-1}=2,
\]
where we used that $\left\{  \overrightarrow{\boldsymbol{k}_{i}}\right\}
_{1\leq i\leq N-1}$ are two by two orthogonal. Finally, the conditions
(\ref{Kineematic_restritions}), (\ref{System_3}) are satisfied since
$\left\langle \overrightarrow{\boldsymbol{k}_{i}}\text{,}\overrightarrow
{\boldsymbol{k}_{j}}\right\rangle =0$ and $-t_{i}t_{j}=-\frac{2}{N-1}$ for
$1\leq i\neq j \leq N-1$.
\end{proof}

Next, for arbitrary $N\geq4$, if (\ref{Eq_solution_set}) admits no solutions
$\boldsymbol{k}_{1},\ldots,\boldsymbol{k}_{N}$ in $\mathbb{C}^{l+1}$, the
relevant question, say for $\mathbb{K}=\mathbb{R}$, becomes: \emph{does the
algebraic set, determined by the kinematic restrictions
(\ref{Kineematic_restritions}), have points outside the polar locus of
$A_{\mathbb{R}}^{(N)}\left(  \boldsymbol{k}\right)  $? }

This is a very subtle problem, for the following reason. Looking at the
description of the polar locus (\ref{polarlocusR}) in Theorem \ref{TheoremC},
the most conceptual strategy is find a solution of
(\ref{Kineematic_restritions}) such that
\[
\sum_{ij}\boldsymbol{k}_{i}\boldsymbol{k}_{j}\notin\mathbb{Z},
\]
for all sums ranging over some non-empty subset of the $\frac{N(N-3)}{2}$
different occurring $ij$ in $A_{\mathbb{R}}^{(N)}\left(  \boldsymbol{k}%
\right)  $, i.e., all $ij$ satisfying $1\leq i<j\leq N-1$ \emph{except} $1$,
$(N-1)$. The subtlety comes from the fact that (\ref{Kineematic_restritions})
implies the relation
\[
\sum_{1\leq i<j\leq N-1}\boldsymbol{k}_{i}\boldsymbol{k}_{j}=-(N-2)\in
\mathbb{Z}.
\]
It is a reasonable idea to search for solutions in $\mathbb{C}^{l+1}$ with $l$
as small as possible. This probably simplifies computations, and such a
solution then induces automatically a solution for larger $l$, simply by
putting all extra coordinates equal to zero.

When $l=1$ however, already for $N=4$ \emph{all} solutions of
(\ref{Kineematic_restritions}) belong to the polar locus, as can be verified
by a short computation. More precisely, for all solutions it turns out that
either $\boldsymbol{k}_{1}\boldsymbol{k}_{2} = -2$, $\boldsymbol{k}%
_{2}\boldsymbol{k}_{3} = -2$, or $\boldsymbol{k}_{1}\boldsymbol{k}_{2}
+\boldsymbol{k}_{2}\boldsymbol{k}_{3} = 0$. Hence, they all belong to the
polar locus (\ref{polarlocusR}) with $t=1$, see Example \ref{Example_N_4} or
Section \ref{Section_Veneziaano_amplitude} below.

When $l=2$, it is conceivable that solutions outside the polar locus exist for
any $N$. We found the following family of solutions of
(\ref{Kineematic_restritions}), that lie outside the polar locus for \lq
many\rq\ $N$. We did not pursue this line of investigation further.

\subsubsection{\label{Lastexample}Example}

We denote as before $\boldsymbol{k}_{i}=\left(  t_{i},\overrightarrow
{\boldsymbol{k}_{i}}\right)  $, where now $t_{i}\in\mathbb{C}$ and
$\overrightarrow{\boldsymbol{k}_{i}}\in\mathbb{C}^{2}$ for $i=1,\ldots,N$. For
$i=1,\dots,N-1$ we take $t_{i}=\frac{\sqrt{-2}}{N-1}$ and $\overrightarrow
{\boldsymbol{k}_{i}}\in\mathbb{R}^{2}$ with $\left\Vert \overrightarrow
{\boldsymbol{k}_{i}}\right\Vert =\frac{\sqrt{2N(N-2)}}{N-1}$, such that these
$\overrightarrow{\boldsymbol{k}_{i}}$ are `equidistributed', i.e., the angle
between $\overrightarrow{\boldsymbol{k}_{i}}$ and $\overrightarrow
{\boldsymbol{k}_{i+1}}$ is always $\frac{2\pi}{N-1}$. Then $\sum_{i=1}%
^{N-1}\overrightarrow{\boldsymbol{k}_{i}}=\overrightarrow{0}$.

Next, we take $t_{N}= -\sqrt{-2}$ and $\overrightarrow{\boldsymbol{k}_{N}} =
\overrightarrow{0} $. Then (\ref{Kineematic_restritions}) is clearly
satisfied. All the $\boldsymbol{k}_{i}\boldsymbol{k}_{j}$ are of the form
\[
\frac{2}{(N-1)^{2}} + \cos\left(  m \frac{2\pi}{N-1}\right)  \frac
{2N(N-2)}{(N-1)^{2}}
\]
for some positive integer $m$, bounded by $\frac{N-2}{2}$.

Using some number theoretic arguments, we verified that this solution is
outside the polar locus of $A_{\mathbb{R}}^{(N)}\left(  \boldsymbol{k}\right)
$ if $N\leq10$, and more generally if for instance $N-1$ is a prime number.

\section{\label{Section_Amplitudes_gammafunctions}Amplitudes and gamma
functions}

\subsection{\label{Section_Veneziaano_amplitude}Veneziano amplitude}

We recall that in the case $N=4$, $K=\mathbb{R}$, $A_{\mathbb{R}}^{(4)}\left(
\boldsymbol{k}\right)  $ is the Veneziano amplitude
\[
A_{\mathbb{R}}^{(4)}\left(  \boldsymbol{k}\right)  =\int_{\mathbb{R}%
}\left\vert x\right\vert ^{\boldsymbol{k}_{1}\boldsymbol{k}_{2}}\left\vert
1-x\right\vert ^{\boldsymbol{k}_{2}\boldsymbol{k}_{3}}dx,
\]
see e.g. \cite{Veneziano}, see also \cite[Section 11]{Broedel et al},
\cite[Chapter 3, Section XIV]{V-V-Z} and the references therein. Here the
momenta $\boldsymbol{k}_{1},\boldsymbol{k}_{2},\boldsymbol{k}_{3}%
,\boldsymbol{k}_{4}\in\mathbb{R}^{l+1}$.\ The standard Mandelstam variables
are defined as%
\[
s=-\left(  \boldsymbol{k}_{1}+\boldsymbol{k}_{2}\right)  ^{2}\text{,\qquad
\ }t=-\left(  \boldsymbol{k}_{2}+\boldsymbol{k}_{3}\right)  ^{2}\text{,}\qquad
u=-\left(  \boldsymbol{k}_{2}+\boldsymbol{k}_{4}\right)  ^{2}.
\]
They satisfy the condition $s+t+u=4m^{2}=-8$. Notice that by using the momenta
conservation condition we have $t=-\left(  \boldsymbol{k}_{1}+\boldsymbol{k}%
_{4}\right)  ^{2}$ and $u=-\left(  \boldsymbol{k}_{1}+\boldsymbol{k}%
_{3}\right)  ^{2}$. We set
\[
\alpha\left(  y\right)  :=1+\frac{1}{2}y
\]
for the standard Regge trajectory. Then%
\begin{multline*}
A_{\mathbb{R}}^{(4)}\left(  \boldsymbol{k}\right)  =\int_{\mathbb{R}%
}\left\vert x\right\vert ^{-\alpha\left(  s\right)  -1}\left\vert
1-x\right\vert ^{-\alpha\left(  t\right)  -1}dx=\\
\frac{\Gamma\left(  -\alpha\left(  s\right)  \right)  \Gamma\left(
-\alpha\left(  t\right)  \right)  }{\Gamma\left(  -\alpha\left(  s\right)
-\alpha\left(  t\right)  \right)  }+\frac{\Gamma\left(  -\alpha\left(
t\right)  \right)  \Gamma\left(  -\alpha\left(  u\right)  \right)  }%
{\Gamma\left(  -\alpha\left(  t\right)  -\alpha\left(  u\right)  \right)
}+\frac{\Gamma\left(  -\alpha\left(  u\right)  \right)  \Gamma\left(
-\alpha\left(  s\right)  \right)  }{\Gamma\left(  -\alpha\left(  u\right)
-\alpha\left(  s\right)  \right)  },
\end{multline*}
where%
\begin{equation}
\Gamma\left(  z\right)  =\int_{\mathbb{R}}x^{z-1}e^{-x}dx,\text{ for
}\operatorname{Re}(z)>0, \label{gamma_function}%
\end{equation}
is the gamma function, which is holomorphic in $\operatorname{Re}(z)>0$. This
function admits a meromorphic continuation to $\mathbb{C}$ with poles in the
non-positive integers. The gamma function has no zeros, hence $\frac{1}%
{\Gamma\left(  z\right)  }$ is an entire function.

By using the fact that the functions $\frac{1}{\Gamma\left(  -\alpha\left(
s\right)  -\alpha\left(  t\right)  \right)  }$, $\frac{1}{\Gamma\left(
-\alpha\left(  t\right)  -\alpha\left(  u\right)  \right)  }$, $\frac
{1}{\Gamma\left(  -\alpha\left(  u\right)  -\alpha\left(  s\right)  \right)
}$ are entire, $A_{\mathbb{R}}^{(4)}\left(  \boldsymbol{k}\right)  $ is
holomorphic in the solution set of
\[
\left\{
\begin{array}
[c]{lll}%
-\alpha\left(  s\right)  =-1+\frac{1}{2}\left(  \boldsymbol{k}_{1}%
+\boldsymbol{k}_{2}\right)  ^{2}>0 & \Leftrightarrow & \boldsymbol{k}%
_{1}\boldsymbol{k}_{2}>-1;\\
&  & \\
-\alpha\left(  t\right)  =-1+\frac{1}{2}\left(  \boldsymbol{k}_{2}%
+\boldsymbol{k}_{3}\right)  ^{2}>0 & \Leftrightarrow & \boldsymbol{k}%
_{2}\boldsymbol{k}_{3}>-1;\\
&  & \\
-\alpha\left(  u\right)  =-1+\frac{1}{2}\left(  \boldsymbol{k}_{2}%
+\boldsymbol{k}_{4}\right)  ^{2}>0 & \Leftrightarrow & \boldsymbol{k}%
_{2}\boldsymbol{k}_{4}>-1.
\end{array}
\right.
\]
On the other hand, by applying Theorem \ref{thirdcondition} and Example
\ref{Example_N_4}, $A_{\mathbb{R}}^{(4)}\left(  \boldsymbol{k}\right)  $ is
holomorphic in the solution set of
\[
\boldsymbol{k}_{1}\boldsymbol{k}_{2}>-1;\qquad\boldsymbol{k}_{2}%
\boldsymbol{k}_{3}>-1;\qquad\boldsymbol{k}_{1}\boldsymbol{k}_{2}%
+\boldsymbol{k}_{2}\boldsymbol{k}_{3}<-1.
\]
By using the momenta conservation condition $\boldsymbol{k}_{2}\boldsymbol{k}%
_{1}+2+\boldsymbol{k}_{2}\boldsymbol{k}_{3}+\boldsymbol{k}_{2}\boldsymbol{k}%
_{4}=0$, we can replace the last condition by
\[
\boldsymbol{k}_{1}\boldsymbol{k}_{2}+\boldsymbol{k}_{2}\boldsymbol{k}%
_{3}=-2-\boldsymbol{k}_{2}\boldsymbol{k}_{4}<-1\Leftrightarrow\boldsymbol{k}%
_{2}\boldsymbol{k}_{4}>-1 ,
\]
which means that Example \ref{Example_N_4} indeed provides the exact domain of
convergence of the Veneziano amplitude $A_{\mathbb{R}}^{(4)}\left(
\boldsymbol{k}\right)  $. (We assumed implicitly that $l\geq2$.)

\subsection{\label{Section_A_N_sum_gammas}$A_{\mathbb{R}}^{(N)}\left(
\boldsymbol{k}\right)  $ as a sum of gamma functions}

In this subsection we show that $A_{\mathbb{R}}^{(N)}\left(  \boldsymbol{k}%
\right)  $ is a combination of gamma functions. Indeed,
\begin{equation}
A_{\mathbb{R}}^{(N)}\left(  \boldsymbol{k}\right)  =%
{\displaystyle\sum\limits_{I\subseteq\left\{  2,\ldots,N-2\right\}  }}
{\displaystyle\sum\limits_{r\in T(I)}}
C_{I,r}\left(  \boldsymbol{k}\right)  \Gamma\left(  \sum_{ij\in M\left(
I\right)  }N_{ij,r}(I)\boldsymbol{k}_{i}\boldsymbol{k}_{j}+\gamma
_{r}(I)\right)  , \label{Formula_amplitude_R}%
\end{equation}
where $N_{ij,r}\left(  I\right)  ,\gamma_{r}\left(  I\right)  \in\mathbb{Z}$,
and $M(I)$, $T(I)$ are finite sets as in Theorem \ref{TheoremC}, and the
$C_{I,r}\left(  \boldsymbol{k}\right)  $ are holomorphic functions. This
formula is a consequence of the fact that the meromorphic continuation of the
zeta function $Z_{\mathbb{R}}^{(N)}\left(  \boldsymbol{s}\right)  $ can be
given in terms of gamma functions. This requires using the Bernstein-Sato
theory, see \cite[Theorem 5.3.1 and 5.4.1]{Igusa}. The techniques used here
are contained in the proof of Theorem 5.4.1 in \cite{Igusa}. The formula
(\ref{Formula_amplitude_R}) is obtained by providing an explicit meromorphic
continuation in terms of gamma functions for the monomial integrals
$J_{\mathbb{R}}(s_{1},\ldots,s_{m})$\ in\ Lemma \ref{Lemma10}. We explain this construction.

Take $f(x)\in\mathbb{R}\left[  x_{1},\ldots,x_{n}\right]  \smallsetminus
\mathbb{R}$. Let $b_{f}(s)$ denote the Bernstein-Sato polynomial of $f$. It is
well known that all roots of this monic polynomial are negative rational
numbers. Writing $b_{f}(s)= {\textstyle\prod\nolimits_{\lambda}}\left(
s+\lambda\right)  $, we introduce $\gamma_{f}\left(  s\right)  :=
{\textstyle\prod\nolimits_{\lambda}}\Gamma\left(  s+\lambda\right)  $, where
$\Gamma$ is the gamma function (\ref{gamma_function}).

Set $V:=\left\{  x\in\mathbb{R}^{n};f(x)>0\right\}  $. Given a Schwartz
function $\Phi\in\mathcal{S}(\mathbb{R}^{n})$, we consider
\[
f_{+}^{s}\left(  \Phi\right)  : =\int_{V}f(x)^{s}\Phi\left(  x\right)
dx\text{ for }\operatorname{Re}(s)>0.
\]
Then it turns out that $f_{+}^{s}\left(  \Phi\right)  =\gamma_{f}\left(
s\right)  B(\Phi,s)$, where $B(\Phi,s)$ is an entire function in $s$.
Furthermore for $s$ fixed, $\Phi\rightarrow B(\Phi,s)$ is a tempered
distribution, cf. \cite[Theorem 5.3.1]{Igusa}.

We use this result in the case in which $f(x)$ is a monomial. We first
consider the integral
\[
I_{\theta}(s):=\int_{\mathbb{R}}\theta\left(  y\right)  \left\vert
y\right\vert ^{Ns+v-1}dy=\int_{0}^{\infty}\left\{  \theta\left(  y\right)
+\theta\left(  -y\right)  \right\}  y^{Ns+v-1}dy,
\]
where $\theta$ is a test function, $N\geq1$, $v\geq1$, and $\operatorname{Re}%
(s)>-\frac{v}{N}$. In this case $f(x)=x$, $b_{f}(s)=\left(  s+1\right)  $, and
$I_{\theta}(s)$ admits a meromorphic continuation of the form
\begin{equation}
I_{\theta}(s)=\Gamma\left(  Ns+v\right)  B(\theta,s), \label{Integral_1}%
\end{equation}
where $B(\theta,s)$ is an entire function.

Next, we consider integrals of type
\begin{align*}
J_{\theta}(s_{1},s_{2})  &  =\int_{\mathbb{R}^{2}}\theta\left(  y_{1}%
,y_{2}\right)  \left\vert y_{1}\right\vert ^{N_{1}s_{1}+v_{1}-1}\left\vert
y_{2}\right\vert ^{N_{2}s_{2}+v_{2}-1}dy_{1}dy_{2}\\
&  =\int_{0}^{\infty}\int_{0}^{\infty}\widetilde{\theta}\left(  y_{1}%
,y_{2}\right)  y_{1}^{N_{1}s_{1}+v_{1}-1}y_{2}^{N_{2}s_{2}+v_{2}-1}%
dy_{1}dy_{2},
\end{align*}
with $N_{1}$, $N_{2}\geq1$, $v_{1}$, $v_{2}\geq1$, and $\operatorname{Re}
(s_{1})>-\frac{v_{1}}{N_{1}}$, $\operatorname{Re}(s_{2})>-\frac{v_{2}}{N_{2}}$.

For a fixed $s_{2}$, the tempered distribution $\theta\rightarrow J_{\theta
}(s_{1},s_{2})$ admits a Laurent expansion around $a \in\mathbb{C}$ of the
form
\[
J_{\theta}(s_{1},s_{2})=%
{\displaystyle\sum\limits_{k_{1}=-d_{1}}}
C_{k_{1}}\left(  \theta,s_{2}\right)  \left(  s_{1}-a\right)  ^{k_{1}},
\]
where $\theta\rightarrow C_{k_{1}}\left(  \theta,s_{2}\right)  $ is a tempered
distribution. Now expanding $C_{k_{1}}\left(  \theta,s_{2}\right)  $ around $b
\in\mathbb{C}$ we get
\[
J_{\theta}(s_{1},s_{2})=%
{\displaystyle\sum\limits_{k_{1}=-d_{1}}}
{\displaystyle\sum\limits_{k_{2}=-d_{2}}}
C_{k_{1},k_{2}}\left(  \theta\right)  \left(  s_{1}-a\right)  ^{k_{1}}\left(
s_{2}-b\right)  ^{k_{2}},
\]
where each $\theta\rightarrow C_{k_{1},k_{2}}\left(  \theta\right)  $ is a
tempered distribution, cf. \cite[pg. 65-67]{Igusa}.

Then, to determine weather or not $C_{k_{1},k_{2}}\neq0$, it is sufficient to
show that $C_{k_{1},k_{2}}\left(  \theta\right)  \neq0$ for $\theta$ in a
dense subset of $\mathcal{S}(\mathbb{R}^{2})$. There is a dense subset in
$\mathcal{S}(\mathbb{R}^{2})$ formed by functions of type $\psi_{1}\left(
y_{1}\right)  \psi_{2}\left(  y_{2}\right)  $, where $\psi_{1}$, $\psi_{2}\in$
$\mathcal{S}(\mathbb{R})$, cf. \cite[Lemma 5.4.2]{Igusa}. Then by
(\ref{Integral_1}),%
\[
J_{\psi_{1}\psi_{2}}(s_{1},s_{2})=\Gamma\left(  N_{1}s_{1}+v_{1}\right)
\Gamma\left(  N_{2}s_{2}+v_{2}\right)  B(\psi_{1}\psi_{2},s_{1},s_{2}),
\]
where $\psi_{1}\psi_{2}\rightarrow B(\psi_{1}\psi_{2},s_{1},s_{2})$ is a
tempered distribution, which is a holomorphic function in $\mathbb{C}^{2}$.
Now given $(s_{1},s_{2})$, and by using the fact that the Schwartz topology is
metrizable, we obtain that the functional $B(\psi_{1}\psi_{2},s_{1},s_{2})$
has a unique extension to $\mathcal{S}(\mathbb{R}^{2})$. This implies that
\begin{equation}
J_{\theta}(s_{1},s_{2})=\Gamma\left(  N_{1}s_{1}+v_{1}\right)  \Gamma\left(
N_{2}s_{2}+v_{2}\right)  B(\theta,s_{1},s_{2}), \label{Integral_2}%
\end{equation}
where $B(\theta,s_{1},s_{2})$ is a holomorphic function in $\mathbb{C}^{2}$.

The above argument can easily be extended to the monomial integrals given in
Lemma \ref{Lemma10} in the case $\mathbb{K}=\mathbb{R}$, $\mathbb{C}$.
Finally, by using the calculations given in Example \ref{Example_N_5}, one
obtains that $Z_{\mathbb{R}}^{(5)}\left(  \boldsymbol{s}\right)  $ is a (very
large) sum of monomial integrals of the type $J_{\mathbb{R}}(s_{1}%
,\ldots,s_{m})$. Then by using the above argument, one obtains a formula of
type (\ref{Formula_amplitude_R}) for $A_{\mathbb{R}}^{(5)}\left(
\boldsymbol{s}\right)  $. We do not include it here due to the length of this formula.

\bigskip

\begin{acknowledgement}
\textrm{The authors wish to thank Hugo Garc\'{\i}a-Compe\'{a}n for some
discussions on string amplitudes, and Erik Panzer and Cl\'{e}ment Dupont for
calling our attention to some recent work on string amplitudes. We wish also
to thank the referees for many important suggestions and questions that helped
us to improve the original manuscript. }
\end{acknowledgement}

\end{document}